\documentclass{vldb}
\usepackage{graphicx}
\usepackage{xcolor}
\usepackage{subfigure}
\usepackage{enumitem}

\usepackage{times}
\usepackage[normalem]{ulem}
\usepackage{xspace}
\usepackage[font={small,it}]{caption}
\usepackage{algorithm}
\usepackage{algorithmic}

\newtheorem{definition}{Definition}[section]

\newtheorem{theorem}{Theorem}[section]

\newtheorem{problem}{Problem}[section]
\newtheorem{example}{Example}[section]
\newtheorem{lemma}{Lemma}[section]

\newtheorem{proof}{Proof}[section]
\newcommand{\qedr}{\hfill $\Box$}

\newcommand{\eg}{\emph{e.g.}\xspace}

\newcommand{\extra}[1]{}
\renewcommand{\paragraph}[1]{\smallskip \noindent {\bf #1}}

\newcommand{\papertext}[1]{}
\newcommand{\techreport}[1]{#1}

\newenvironment{denselist}{
    \begin{list}{\small{$\bullet$}}%
    {\setlength{\itemsep}{0ex} \setlength{\topsep}{0ex}
    \setlength{\parsep}{0pt} \setlength{\itemindent}{0pt}
    \setlength{\leftmargin}{1.5em}
    \setlength{\partopsep}{0pt}}}%
    {\end{list}}

\makeatletter
\def\@copyrightspace{\relax}
\makeatother

\begin{document}

%\title{Skipping Out of Local Optima: Globally Optimal Crowdsourcing Quality Management}
\title{Globally Optimal Crowdsourcing Quality Management}
%\\ \Large{(Placeholder. Full technical report will appear on December 7, 2014.)}
\numberofauthors{3}
\author{
\alignauthor Akash Das Sarma\\
       \affaddr{Stanford University}\\
       \email{akashds@stanford.edu}
\alignauthor Aditya G. Parameswaran \\
        \affaddr{University of Illinois (UIUC)} \\
        \email{adityagp@illinois.edu}
\alignauthor Jennifer Widom \\
        \affaddr{Stanford University}\\
        \email{widom@cs.stanford.edu}
}

\maketitle

%\frenchspacing

\begin{abstract}
%We consider the problem of optimally filtering (or rating) a set of items based on predicates (or scoring) requiring human evaluation. Filtering and rating are ubiquitous problems across crowdsourcing applications. We consider the setting where we are given a set of items and a set of worker responses for each item: yes/no in the case of filtering and an integer value in the case of rating. We assume that items have a true inherent value that is unknown, and workers draw their responses from a common, but hidden, error distribution. Our goal is to simultaneously assign a ground truth to the item-set and estimate the worker error distribution. Previous work in this area (\cite{raykar2011ranking,whitehill2009whose}) has focused on heuristics such as Expectation Maximization (EM), providing only a local optima guarantee, while we have developed a general framework that finds a maximum likelihood solution. Our approach extends to a number of variations on the filtering and rating problems.

We study {\em crowdsourcing quality management}, that is, given worker responses to a set of tasks, our goal is to jointly estimate the true answers for the tasks, as well as the quality of the workers. Prior work on this problem relies primarily on applying Expectation-Maximization (EM) on the underlying maximum likelihood problem to estimate true answers as well as worker quality. Unfortunately, EM only provides a locally optimal solution rather than a globally optimal one. Other solutions to the problem (that do not leverage EM) fail to provide global optimality guarantees as well.

In this paper, we focus on filtering, where tasks require the evaluation of a yes/no predicate, and rating, where tasks elicit integer scores from a finite domain. We design algorithms for finding the global optimal estimates of correct task answers and worker quality for the underlying maximum likelihood problem, and characterize the complexity of these algorithms.  Our algorithms conceptually consider all mappings from tasks to true answers (typically a very large number), leveraging two key ideas to reduce, by several orders of magnitude, the number of mappings under consideration, while
preserving optimality. We also demonstrate that these algorithms often find more accurate estimates than EM-based algorithms. This paper makes an important contribution towards understanding the inherent complexity of globally optimal crowdsourcing quality management.

%{\fontsize{10pt}{15pt} \selectfont Testing font size...}
\end{abstract}

%!TEX root=Parameter Estimation.tex

\section{Introduction}
%\subsection{Problem Overview}
\label{intro}
%First paragraph: (Why are we focusing on quality assurance)
Crowdsourcing~\cite{doan2011crowdsourcing} enables data scientists to collect human-labeled data
at scale for machine learning algorithms, including those involving
image, video, or text analysis.
However, human workers often make mistakes while answering these tasks.
Thus, {\em crowdsourcing quality management}, i.e., jointly
estimating human worker quality as well as answer quality---the probability of different answers for the tasks---is essential.
While knowing the answer quality helps us with the set of tasks at hand,
knowing the quality of workers helps us estimate
the true answers for future tasks, and in deciding whether to hire or fire specific workers.

In this paper, we focus on {\em rating tasks}, i.e., those where the answer is one
from a fixed set of ratings $\in \{1, 2, \ldots, R\}$.
This includes, as a special case, {\em filtering tasks},  where
the ratings are binary, i.e., $\{0, 1\}$.
Consider the following example:
say a data scientist intends to design a sentiment analysis algorithm for tweets.
To train such an algorithm, she needs a training dataset of tweets, rated on sentiment.
Each tweet needs to be rated on a scale of $\{1, 2, 3\}$, where $1$ is
negative, $2$ is neutral, and $3$ is positive.
A natural way to do this is to display each tweet, or item,
to human workers hired via a crowdsourcing marketplace
like Amazon's Mechanical Turk \cite{mturk}, and have workers rate each item on sentiment from 1---3.
Since workers may answer these rating tasks incorrectly,
we may have multiple workers rate each item.
Our goal is then to jointly estimate sentiment of each tweet
and the accuracy of the workers.

%Next paragraph: (What's wrong with existing work) - TO WRITE
%* Argue that EM-based schemes, which is the predominant de-facto crowdsourcing quality assurance scheme, has no guarantees.
%* Often, it gets stuck in local optima, depending on the initialization (as we will show in experiments -- have forward pointer to experimental section).
%* Other techniques for optimal quality assurance, e.g., karger, jordan, etc. are not truly optimal in that they only give probabilistic guarantees (BTW, we will need to compare with at least one of these in our experiments.)
%* Another forward pointer: as we will show in this paper, none of these techniques give us truly global optimality.
Standard techniques for solving this estimation problem
typically involve the use of the Expectation-Maximization (EM).
Applications of EM, however, provide no theoretical guarantees.
Furthermore, as we will show in this paper, EM-based algorithms
are highly dependent on
initialization parameters and can often get stuck in undesirable local optima.
Other techniques for optimal quality assurance, some specific to only filtering~\cite{dalvi2013aggregating,GKM,karger2011iterative}, are not provably optimal either,
in that they only give bounds on the errors of their estimates,
and do not provide the globally optimal quality estimates.
We cover other related work in the next section.

In this paper, we present a technique for globally optimal quality management,
that is, finding the maximum likelihood item (tweet) ratings, and worker quality estimates.
If we have $500$ tweets and $3$ possible ratings,
the total number of mappings from tweets to ratings is $3^{500}$.
A straightforward technique for globally optimal quality management is
to simply consider all possible mappings, and for each mapping, infer the overall
likelihood of that mapping. (It can be shown that the best worker error rates
are easy to determine once the mapping is fixed.)
The mapping with the highest likelihood is then the global optimum.

However, the number of mappings even in this simple example, $3^{500}$, is very large,
therefore making this approach infeasible.
Now, for illustration,
let us assume that workers are indistinguishable,
and they all have the same quality (which is unknown).
It is well-understood that at least on Mechanical Turk,
the worker pool is constantly in flux, and it is often hard to
find workers who have attempted enough tasks in order to get robust estimates of worker quality.
(Our techniques also apply to a generalization of this case.)

To reduce this exponential complexity, we use two simple,
but powerful ideas to greatly prune the set of mappings that need to be considered,
from $3^{500}$, to a much more manageable number.
Suppose we have 3 ratings for each tweet---a common strategy
in crowdsourcing is to get a small, fixed number of answers for each question.
 First, we hash ``similar'' tweets that receive the same set of worker ratings into a common bucket.
 As shown in Figure~\ref{fig:intro-ex}, suppose that 300 items each receive three ratings of $3$ (positive), 100 items each receive one rating of 1, one rating of 2 and one rating of 3, and 100 items each receive three ratings of $1$.
 That is, we have three buckets of items, corresponding to the worker answer sets $B_1=\{3,3,3\}$, $B_2=\{1,2,3\}$, and $B_3=\{1,1,1\}$. We now exploit the intuition that
 if two items receive the same set of worker responses they should be treated identically.
 We therefore only consider mappings that assign the same rating to all items (tweets) within a bucket.
 Now, since in our example, we only have 3 buckets each of which can be assigned a rating of 1,2, or 3,
 we are left with just $3^3=27$ mappings to consider.

\begin{figure}
\centering
\includegraphics[scale=0.4]{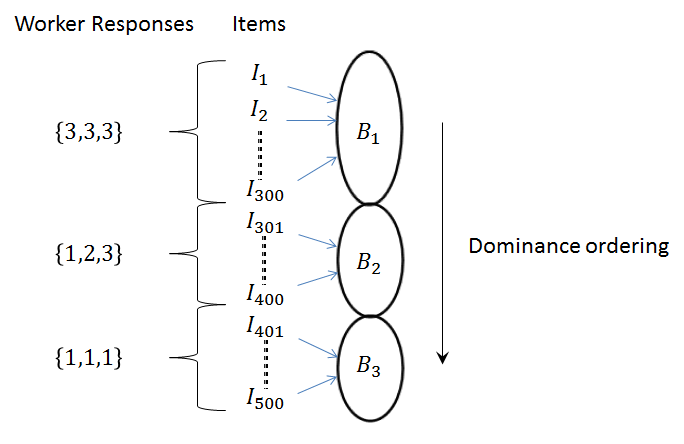}
\vspace{-15pt}
\caption{Mapping Example}
\vspace{-15pt}
\label{fig:intro-ex}
\end{figure}

Next, we impose a partial ordering on the set of buckets
based on our belief that items in certain buckets
should receive a higher final rating than items in other buckets.
Our intuition is simple: if an item received ``higher'' worker ratings than another item,
then its final assigned rating should be higher as well.
In this example, our partial ordering, or {\em dominance ordering} on the buckets is $B_1\geq B_2\geq B_3$, that is, intuitively items which received all three worker ratings
of $3$ should not have a true rating smaller than items in the second or third buckets where items receive lower ratings. This means that we can further reduce our space of $27$ remaining mappings by removing all those mappings that do not respect this partial ordering. The number of such remaining mappings is 10, corresponding to when all items in the buckets $(B_1,B_2,B_3)$ are mapped respectively to ratings $(3,3,3)$, $(3,3,2)$, $(3,3,1)$, $(3,2,2)$, $(3,2,1)$, $(3,1,1)$, $(2,2,2)$, $(2,2,1)$, $(2,1,1)$, and $(1,1,1)$.

%\sout{To summarize, we restrict the space of all mappings to those
%that assign the same rating to all items in a bucket while respecting
%the partial order on buckets; that is, all items in the same bucket get a common rating,
%and if a bucket is higher than another in the dominance ordering,
%then so is the rating assigned to its items. }

In this paper, we formally show that restricting the mappings in
 this way does not take away from the optimality of the
 solution; i.e., there exists a mapping with the highest
 likelihood that obeys the property that all items
 with same scores are mapped to the same rating,
and at the same time obeys the dominance ordering relationship as described above.

%\agp{fill up}

Our list of contributions are as follows:
\begin{denselist}
\item We develop an intuitive algorithm based on simple,
but key insights that finds a provably optimal maximum likelihood solution
to the problem of jointly estimating true item labels and worker error behavior
for crowdsourced filtering and rating tasks.
Our approach involves reducing the space of potential ground truth mappings
while preserving optimality, enabling an
exhaustive search on an otherwise prohibitive domain.
%\item We develop a novel pruning approach that reduces the space of potential ground truth mappings enabling an exhaustive evaluation and search for a maximum likelihood solution. Our approach does not rely on any machine learning or approximate methods and therefore results in a true optimal solution as opposed to providing probabilistic guarantees.
\item Although we primarily focus on and initially derive
our optimality results for the setting where workers independently
draw their responses from a common error distribution,
we also propose generalizations to harder settings,
for instance, when workers are known to come from distinct classes with separate error distributions.
That said, even the former setting is commonly used in practice, and
represents a significant first step
towards understanding the nature and complexity
of exact globally optimal solutions for this joint estimation problem.

%We initially assume that workers are fleeting and do not have persistent identities. Under this setting, we model the worker population as a group of independent workers, all drawing their responses from a common response distribution. Since this is a strong assumption, we also propose generalizations to our algorithm when workers are known to come from distinct classes, each with separate response distributions. While our approach can theoretically be extended to the limiting case where each worker has an independent and different distribution, it is impractical to do so. We leave the problem of leveraging our ideas to find globally maximum likelihood solutions to this harder setting as future work.

\item We perform experiments on synthetic and real datasets
to evaluate the performance of our algorithm on a variety of different metrics.
Though we optimize for likelihood, we also test
the accuracy of predicted item labels and worker response distributions.
We show that our algorithm also does well on these other metrics. We test our algorithm on a real dataset where our assumptions about the worker model do not necessarily hold, and show that our algorithm still yields good results.
\end{denselist}

\subsection{Related Literature}
\label{sec:related}
%\smallskip
%\noindent {\bf Applications.}
%\agp{the following paragraph feels out of place because you go back to worker quality estimation in the next paragraph.  I think you need a better organization than what you have right now. Consider splitting it into EM-Based Joint Estimation Techniques, Other Worker Quality Estimation Techniques, Worker Modeling(?), Applications}

Crowdsourcing is gaining importance as a platform for a variety of different applications where automated machine learning techniques don't always perform well, %. For example, some important applications that have been studied recently include
\eg, filtering~\cite{crowdscreen} or labeling~\cite{DBLP:journals/jmlr/RaykarY12,get-another-label,welinder2010multidimensional} of text, images, or video, and entity resolution \cite{entity-matching,DBLP:journals/pvldb/VesdapuntBD14,crowder,steven-iterative}.
One crucial problem in crowdsourced applications is that of worker quality: since human workers often make mistakes, it is important to model and characterize their behavior in order to aggregate high-quality answers.
 % For labeling applications, there has been quite a bit of work on modeling tasks and worker qualities \cite{whitehill-accuracy,raykar-whom-to-trust}. %\agp{for what?}

\smallskip
\noindent {\bf EM-based joint estimation techniques.}
We study the particular problem of jointly estimating hidden values (item ratings) and a related latent set of parameters (worker error rates) given a set of observed data (worker responses).
% has been studied both in a machine learning context, as well as a crowdsourcing one. %We discuss the previous works relevant to this problem under these two broad categories.
A standard machine learning technique for estimating parameters with unobserved latent variables is Expectation Maximization~\cite{GuptaEMSurvey}. There has been significant work in using EM-based techniques to estimate true item values and worker error rates, such as \cite{dawid-skene,raykar-whom-to-trust,whitehill2009whose}, and subsequent modifications using Bayesian techniques~\cite{carpenter2011hierarchical,donmez-learning-accuracy}. In \cite{raykar2010learning}, the authors use a supervised learning approach to learn a classifier and the ground truth labels simultaneously. In general, these machine learning based techniques only provide probabilistic guarantees and cannot ensure optimality of the estimates. We solve the problem of finding a global, provably maximum likelihood solution for both the item values and the worker error rates. That said, our worker error model is simpler than the models considered in these papers---in particular, we do not consider worker identities or difficulties of individual items. While we do provide generalizations to our approach that relax some of these assumptions, they can be inefficient in practice. However, we study this simpler model in more depth, providing optimality guarantees. Since our work represents the first providing optimality guarantees (even for a restricted setting), it represents an important step forward in our understanding of crowdsourcing quality management.  Furthermore, anecdotally, even the simpler model is commonly used for platforms like Mechanical Turk, where the workers are fleeting.

%\agp{we do need to acknowledge that our worker error model is weaker than theirs.
%Maybe something like: That said, our worker error model is simpler than the models considered in XXX---in particular,
%we do not consider worker identities, difficulties of individual items, or XXX. However, we study
%this simpler model in more depth, providing optimality guarantees. Furthermore,
%anecdotally, even the simpler model is commonly used for platforms like Mechanical Turk.}
%\smallskip
%\noindent {\bf Crowd Algorithms.}

%The problem of evaluating binary tasks (filtering) and estimating worker error rates for such tasks is in particular relevant to our problem.

\smallskip
\noindent {\bf Other techniques with no guarantees.}
There has been some work that adapts techniques different from EM to solve the problem
of worker quality estimation. For instance, Chen et al.~\cite{chen2013optimistic} adopts
approximate Markov Decision Processes to perform simultaneous worker quality estimation and budget
allocation. Liu et al.~\cite{LPI} uses variational inference for worker quality management
on filtering tasks (in our case, our techniques apply to both filtering and rating).
Like EM-based techniques, these papers do not provide any theoretical guarantees.

\smallskip
\noindent {\bf Weaker guarantees.} There has been a lot of recent work on providing
partial probabilistic guarantees or asymptotic guarantees on accuracies of answers or worker estimates, for various
problem settings and assumptions, and using various techniques.
We first describe the problem settings adopted by these papers, then their solution techniques, and then describe
their partial guarantees.

The most general problem setting adopted by these papers is identical to us (i.e., rating tasks with arbitrary
bipartite graphs connecting workers and tasks)~\cite{zhang2014spectral}; most papers
focus only on filtering~\cite{dalvi2013aggregating,GKM,karger2011iterative}, or operate only
when the graph is assumed to be randomly generated~\cite{KOS,karger2011iterative}. Furthermore, most of these
papers assume that the false positive and false negative rates are the same.

The papers draw from various techniques, including just spectral methods~\cite{dalvi2013aggregating,GKM},
just message passing~\cite{karger2011iterative}, a combination of spectral methods and message passing~\cite{KOS},
or a combination of spectral methods and EM~\cite{zhang2014spectral}.

In terms of guarantees, most of the papers provide probabilistic bounds~\cite{dalvi2013aggregating,KOS,karger2011iterative,zhang2014spectral}, while some only provide asymptotic bounds~\cite{GKM}. For example, Dalvi et al.~\cite{dalvi2013aggregating}, which represents the state of the art over multiple papers~\cite{GKM,KOS,karger2011iterative},
show that under certain assumptions about the graph structure (depending on the eigenvalues)
the error in their estimates of worker quality is lower than some quantity with probability greater than $1-\delta$.

Thus, overall, all the work discussed so far provides probabilistic guarantees on their item value predictions, and error bound guarantees on their estimated worker qualities. In contrast, we consider the problem of finding a global maximum likelihood estimate for the correct answers to tasks and the worker error rates.

\smallskip
\noindent {\bf Other related papers.} Joglekar et al.~\cite{joglekar2013evaluating} consider the problem of finding confidence bounds on worker error rates. Our paper is complementary to theirs in that, while they solve the problem of obtaining confidence bounds on the worker error rates, we consider the problem of finding the maximum likelihood estimates to the item ground truth and worker error rates.

Zhou et al.~\cite{zhou2012learning,zhouaggregating} use minimax entropy to perform worker
quality estimation as well as inherent item difficulty estimation;
here the inherent item difficulty is represented as a vector. Their technique only applies
when the number of workers attempting each task is very large;
here, overfitting (given the large number of hidden parameters) is no longer a concern.
For cases where the number of workers attempting each task is in the order of 50 or 100
(highly unrealistic in practical applications), the authors demonstrate that the
scheme outperforms vanilla EM.

\smallskip
\noindent
{\bf Summary.} In summary, at the highest level, our work differs from all previous
work in its focus on finding a globally optimal solution to the maximum likelihood problem.
We focus on a simpler setting, but do so in more depth, representing
significant progress in our understanding of global optimality.
Our globally optimal solution
uses simple and intuitive insights to reduce the search space of
possible ground truths, enabling exhaustive evaluation. Our general
framework leaves room for further study and has the potential for more sophisticated algorithms that build on our reduced space.

\section{Preliminaries}
\label{prelims}

We start by introducing some notation, and then describe
the general problem that we study in this paper;
specific variants will be considered in subsequent sections.

\smallskip
\noindent {\bf Items and Rating Questions.}
We let ${\bf I}$ be a set of $|\mathbf{I}|=n$ items.
Items could be, for example, images, videos, or pieces of text.%; we focus on images to motivate the rest of the notation.

Given an item $I \in {\bf I}$, we can ask a worker $w$ to answer
{\em a rating question} on that item.
That is, we ask a worker:
{\em What is your rating $r$ for the item $I$?}.
We allow workers to rate the item with any value  $\in \{1, 2, \ldots, R\}$.

\begin{example}
\label{ex:prelim}
Recall our example application from Section \ref{intro}, where we have $R=3$ and workers can rate tweets
as being negative ($r=1$), neutral ($r=2$), or positive ($r=3$). Suppose we have two items, $\mathbf{I}=\{I_1,I_2\}$ where $I_1$ is positive, or has a true rating of 3, and $I_2$ is neutral, or has a true rating of 2.
\end{example}

\smallskip
\noindent {\bf Response Set.}
We assume that each item $I$ is shown to $m$ %distinct randomly chosen
arbitrary workers,
and therefore receives $m$ ratings $\in [1, R]$.
We denote the set of ratings given by workers for item $I$ as $M(I)$ and write $M(I)=(v_R,v_{R-1},\ldots,v_1)$ if item $I\in\mathbf{I}$ receives $v_i$ responses of rating ``$i$'' across workers, $1\leq i\leq R$. Thus, $\overset{R}{\underset{i=1}{\sum}} v_i = m$.
We call $M(I)$ the {\em response set of $I$}, and $M$ the {\em worker response set} in general.

Continuing with Example \ref{ex:prelim}, suppose we have $m=2$ workers rating each item on the scale of $\{1,2,3\}$. Let $I_1$ receive one worker response of $2$ and one worker response of $3$. Then, we write $M(I_1)=(1,1,0)$. Similarly, if $I_2$ receives one response of $3$ and one response of $1$, then we have $M(I_2)=(1,0,1)$.

\smallskip
\noindent {\bf Modeling Worker Errors.}
We assume that every item $I \in {\bf I}$ has a {\em true rating}
in $[1, R]$ that is not known to us in advance.
What we can do is estimate the true rating using the worker response set.
To estimate the true rating, we need to be able to estimate the probabilities
of worker errors.

We assume every worker draws their responses from a common (discrete)
{\em response probability matrix}, $p$, of size $R \times R$.
Thus, $p(i,j)$ is the probability that a worker rates an item with true value $j$ as having rating $i$.
Consider the following response probability matrix of the workers described in our example ($R=3$):
 $$p=
 \begin{bmatrix}
 0.7 & 0.1 & 0.2\\
 0.2 & 0.8 & 0.2\\
 0.1 & 0.1 & 0.6
 \end{bmatrix}
 $$ H
 ere, the $j^{th}$ column represents the different probabilities of worker responses when an item's true rating is $j$. %, the second column represents probabilities when an item's true rating is $2$, and so on.
 Correspondingly, the $i^{th}$ row represents the probabilities that a worker will rate an item as $i$.
We have $p(1,1)=0.7, p(2,1)=0.2, p(3,1)=0.1$ meaning that given an item whose true rating is 1, workers will rate the item correctly
 with probability $0.7$, give it a rating of 2 with probability 0.2, and give it a rating of 3 with probability 0.1.
The matrix $p$ is in general not known to us.
We aim to estimate both $p$ and the true ratings of items in ${\bf I}$
as part of our computation.

Note that we assume that every response to a rating question
returned by every worker
is independently and identically drawn from this matrix:
thus, each worker responds to each rating question independently of
other questions they may have answered, and other ratings for
the same question given by other workers; and furthermore, all the workers have the same response matrix. In our example, we assume that all four responses (2 responses to each of $I_1,I_2$) are drawn from this distribution.
We recognize that assuming the same response matrix is somewhat stringent---we will consider
generalizations to the case where we can categorize workers
into classes (each with the same response matrix) in Section~\ref{sec:extensions}.
That said, while our techniques can still indeed be applied when
there are a large number of workers or worker classes with distinct response matrices,
it may be impractical.
Since our focus is on understanding the theoretical limits
of global optimality for a simple case, we defer to
future work fully generalizing our techniques to apply to workers with
distinct response matrices, or when worker answers are not independent of each other.

\smallskip
\noindent {\bf Mapping and Likelihood}
We call a function $f:\mathbf{I}\rightarrow\{1,2,\ldots,R\}$ that assigns
ratings to items a \emph{mapping}. The set of actual ratings of items is also a mapping. We call that the {\em ground truth mapping}, $T$. For Example \ref{ex:prelim}, $T(I_1)=3,T(I_2)=2$.
%The function that assigns to each item its true rating is called the {\em ground truth mapping}.

Our goal is to find the \emph{most likely} mapping
and worker response matrix given the response set $M$.
We let the probability of a specific mapping, $f$,
being the ground truth mapping,
given the worker response set $M$ and response probability matrix $p$
be denoted by $\Pr(f|M,p)$.
Using Bayes rule, we have $\Pr(f|M,p)=k\Pr(M|f,p)$, where $\Pr(M|f,p)$ is the probability of seeing worker response set $M$ given that $f$ is the ground truth mapping and $p$ is the true worker error matrix. Here, $k$ is the constant given by $k=\frac{\Pr(f)}{\Pr(M)}$, where $\Pr(f)$ is the (constant) apriori probability of being the ground truth mapping and $\Pr(M)$ is the (constant) apriori probability of seeing worker response set $M$. Thus, $\Pr(M|f,p)$ is the probability of workers providing the responses in $M$, had $f$ been the ground truth mapping and $p$ been the true worker error matrix. We call this value the \emph{likelihood} of the mapping-matrix pair, $f,p$.

We illustrate this concept on our example. We have $M(I_1)=(1,1,0)$ and $M(I_2)=(1,0,1)$. Let us compute the likelihood of the pair $f,p$ when $f(I_1)=2,f(I_2)=2$ and
$p$ is the matrix displayed above. We have
$$\Pr(M|f,p)=\Pr(M(I_1)|f,p)\Pr(M(I_2)|f,p)$$ assuming that rating questions on items are answered independently. The quantity $\Pr(M(I_1)|f,p)$ is the probability that workers drawing their responses from $p$ respond with $M(I_1)$ to an item with true rating $f(I_1)$. Again, assuming independence of worker responses, this quantity can be written as the product of the probability of seeing each of the responses that $I_1$ receives. If $f$ is given as the ground truth mapping, we know that the probability of receiving a response of $i$ is $p(i,f(I_1))=p(i,2)$. Therefore, the probability of seeing $M(1,1,0)$, that is one response of $3$ and one response of $2$, is $p(3,2)p(2,2)=0.1\times 0.8=0.08$. Similarly, $\Pr(M(I_2)|f,p)=p(3,2)\times p(1,2)=0.1\times 0.1=0.01$. Combining all of these expressions, we have $\Pr(M|f,p)=0.01\times 0.08=8\times10^{-4}$.
\begin{table}
\vspace{5pt}
\centering
\small
\begin{tabular}{ | c | c | }
\hline
Symbol & Explanation \\ \hline \hline
$\mathbf{I}$ & Set of items \\ \hline
$M$ & Items-workers response set \\ \hline
$f$ & Items-values mapping  \\ \hline
$p$ & Worker response probability matrix\\ \hline
$\Pr(M|f,p)$ & Likelihood of $(f,p)$ \\ \hline
$m$ & Number of worker responses per item \\ \hline
$T$ & Ground truth mapping \\ \hline
\end{tabular}
\vspace{-5pt}
\caption{Notation Table}
\label{notation_table}
\vspace{-15pt}
\end{table}
%That is, we wish to simultaneously estimate the true values of each item as well as the worker error distribution $p$, such that the probability of seeing the response set $M$ under these assignments is maximized. Consider the filtering problem. We can compute the likelihood of any {\em mapping} $f:\mathbf{I}\rightarrow \{0,1\}$ and a worker error distribution $p$, $\Pr(M|f,p)$ (formalized in Section \ref{formal}), given the response set $M$.
%\smallskip
%\noindent {\bf Maximum Likelihood Problem}
Thus, our goal can be restated as:
\begin{problem}[Maximum Likelihood Problem]
Given $M,\mathbf{I}$, find $$\underset{f,p}{\arg\!\max}\Pr(M|f,p)$$
\end{problem}
A naive solution would be to look at every possible mapping $f'$, compute $p'=\underset{p}\arg\!\max \Pr(M|f',p)$ and choose the $f'$ maximizing the likelihood value $\Pr(M|f',p')$. The number of such mappings, $R^{|\mathbf{I}|}$, is however exponentially large.

We list our notation in Table \ref{notation_table} for ready reference.

\section{Filtering Problem}
\label{sec:filtering}
Filtering can be regarded as a special case of rating where $R=2$. We discuss it separately, first, because its analysis is significantly simpler, and at the same time provides useful insights that we then build upon for the generalization to rating, that is, to the case where $R>2$. For example, consider the filtering task of finding all images of Barack Obama from a given set of images. For each image, we ask workers the question ``is this a picture of Barack Obama''. Images correspond to items and the question ``is this a picture of Barack Obama'' corresponds to the filtering task%\agp{have you defined this before-- fixed}
on each item.
We can represent an answer of ``no'' to the question above
by a score 0, and an answer of ``yes'' by a score 1.
Each item $I\in\mathbf{I}$ now has an
inherent true value in $\{0,1\}$ where a true value of
1 means that the item is one that satisfies the filter,
in this case, the image is one of Barack Obama.
Mappings here are functions $f:\mathbf{I}\rightarrow\{0,1\}$.

Next, we formalize the filtering problem in Section \ref{formal}, describe our algorithm in Section \ref{algo}, prove a maximum likelihood result in Section \ref{sec:filter_proof} and evaluate our algorithm in Section \ref{sec:filter_exp}.% against other heuristics using simulations.

\subsection{Formalization}
\label{formal}
%Recall the notation (Table \ref{notation_table}) introduced in Section \ref{overview}.

%Given evidence $M$ on set $\mathbf{I}$, the selectivity and error rates corresponding to any mapping function $f$ can be easily calculated. If $f$ is the correct mapping, selectivity $s$ is just the fraction of items that $f$ maps to $1$.
Given the response set $M$, we wish to find the maximum likelihood mapping $f:\mathbf{I}\rightarrow\{0,1\}$ and $2\times 2$ response probability matrix, $p$. For the filtering problem, each item has an inherent true value of either 0 or 1, and sees $m$ responses of 0 or 1 from different workers. If item $I$ receives $m-j$ responses of 1 and $j$ responses of 0, we can represent its response set with the tuple or pair $M(I)=(m-j,j)$.

Consider a worker response probability matrix of $p=
 \begin{bmatrix}
 0.7 & 0.2\\
 0.3 & 0.8
 \end{bmatrix}
 $. The first column represents the probabilities of worker responses when an item's true rating is $0$ and the second column represents probabilities when an item's true rating is $1$.
 % Correspondingly, the first row represents the probabilities that a worker will rate an item as $0$ and the second row represents the probabilities that a worker will rate an item as $1$.
 Given that all workers have the same response probabilities, we can characterize their response matrix by just the corresponding worker false positive (FP) and false negative (FN) rates, $e_0$ and $e_1$. That is, $e_0=p(1,0)$ is the probability that a worker responds $1$ to an item whose true value is $0$, and $e_1=p(0,1)$ is the probability that a worker responds $0$ to an item whose true value is $1$. We have $p(1,1)=1-e_1$ and $p(0,0)=1-e_0$. Here, we can describe the entire matrix $p$ with just the two values, $e_0=0.3$ and $e_1=0.2$.

\smallskip
\noindent {\bf Filtering Estimation Problem.}
Let $M$ be the observed response set %\agp{evidence or responses matrix -- fixed}
on item-set $\mathbf{I}$. Our goal is to find %$s^*,e_0^*,e_1^*=\text{Params}(f^*,M)$, where $f^*=\argmax_{f\in\mathbf{F}}\Pr(f=T|M)$.
$$f^*,e_0^*,e_1^*=\underset{f,e_0,e_1}{\arg\!\max} \Pr(M|f,e_0,e_1)$$
Here, $\Pr(M|f,e_0,e_1)$ is the probability of getting the response set $M$, given that $f$ is the ground truth mapping and the true worker response matrix is defined by $e_0,e_1$.

\smallskip
\noindent {\bf Dependance of Response Probability Matrices on Mappings.}
Due to the probabilistic nature of our workers, for a fixed ground truth mapping $T$, %\agp{have you defined this -- fixed}
different worker error rates, $e_0$ and $e_1$ can produce the same response set $M$. These different worker error rates, however have varying likelihoods of occurrence. This leads us to observe that worker error rates ($e_0,e_1$) and mapping functions ($f$) are not independent and are related through any given $M$. In fact, we show that for the maximum likelihood estimation problem, fixing a mapping $f$ enforces a maximum likelihood choice of $e_0,e_1$. We leverage this fact to simplify our problem from searching for the maximum likelihood tuple $f,e_0,e_1$ to just searching for the maximum likelihood mapping, $f$. Given a response set $M$ and a mapping, $f$, we call this maximum likelihood choice of $e_0,e_1$ as the parameter set of $f,M$, and represent it as $\text{Params}(f,M)$. The choice of $\text{Params}(f,M)$ is very intuitive and simple. We show that we just can estimate $e_0$ as the fraction of times a worker disagreed with $f$ on an item $I$ in $M$ when $f(I)=0$, and correspondingly, $e_1$ as the fraction of times a worker responded $0$ to an item $I$, when $f(I)=1$. Under this constraint, we can prove that our original estimation problem, $$\underset{f,e_0,e_1}{\arg\!\max} \Pr(M|f,e_0,e_1)$$ simplifies to that of finding $$\underset{f}{\arg\!\max} \Pr(M|f,e_0^*,e_1^*)$$ where $e_0^*,e_1^*$ are the constants given by $\text{Params}(f,M)$.

\noindent
\begin{example}%[Filter Example]
\label{ex:filter}
Suppose we are given $4$ items $\mathbf{I}=\{I_1,I_2,I_3,I_4\}$ with ground truth mapping $T=(T(I_1),T(I_2),T(I_3),T(I_4))=(1,0,1,1)$. %The selectivity of the filter predicate with respect to this item set is $\frac{3}{4}=0.75$.

Suppose we ask $m=3$ workers to evaluate each item and receive the following number of (``1'',``0'') responses for each respective item: $M(I_1)=(3,0),M(I_2)=(1,2),M(I_3)=(2,1),M(I_4)=(2,1)$. Then, we can evaluate our worker false positive and false negative rates as described above: $e_0=\frac{0+1+1}{3+3+3}=\frac{2}{9}$ (from items $I_1$,$I_3$ and $I_4$) and $e_1=\frac{1}{3}$ (from $I_2$).
\end{example}

We shall henceforth refer to $\Pr(M|f)=\Pr(M|f,\text{Params}(f,M))$ as the \emph{likelihood of a mapping} $f$. For now, we focus on the problem of finding the maximum likelihood mapping with the understanding that finding the error rates is straightforward given the mapping is fixed. In Section \ref{sec:filter_correspondence}, %Lemma \ref{lemma:likelihood},
we formally show that the problem of jointly finding the maximum likelihood response matrix and mapping can be solved by just finding the most likely mapping $f^*$. The most likely triple $f,e_0,e_1$ is then given by $f^*,\text{Params}(f^*,M)$.

It is easy to calculate the likelihood of a given mapping. We have $\Pr(M|f)=\underset{I\in \mathbf{I}}\prod \Pr(M(I)|f,e_0,e_1)$, where $e_0,e_1=\text{Params}(f,M)$ and $\Pr(M(I)|f,e_0,e_1)$ is the probability of seeing the response set $M(I)$ on an item $I\in\mathbf{I}$.
Say $M(I)=(m-j,j)$. Then, we have
\[
\Pr(M(I)|f)=
\begin{cases}
(1-e_1)^{m-j}e_1^j & \text{for } f(I)=1\\
e_0^{m-j}(1-e_0)^j & \text{for } f(I)=0\\
\end{cases}
\]
This can be evaluated in $O(m)$ for each $I\in\mathbf{I}$ by doing one pass over $M(I)$. Thus, $\Pr(M|f)=\prod_{I\in \mathbf{I}} \Pr(M(I)|f)$ can be evaluated in $O(m|\mathbf{I}|)$. We use this as a building block in our algorithm below.

\subsection{Globally Optimal Algorithm}
\label{algo}
%\ads{think of a short but better algo name?}
In this section, we describe our algorithm for finding the maximum likelihood mapping, given a response set $M$ on an item set $\mathbf{I}$.
%We wish to compute the selectivity and error rates for the filter corresponding to the most likely mapping $f:\mathbf{I}\rightarrow V$ given the evidence $M$. Therefore, our problem is to return
%\[
%\argmax_{s,e_0,e_1}(\max_{f\in F}(\prod_{i\in \mathbf{I}} \Pr(f(i)=T(i)|M,s,e_0,e_1)))
%\]
%Using our correspondence between mappings and worker response matrices,
 A naive algorithm could be to scan all possible mappings, $f$, calculating for each, $e_0^*,e_1^*=\text{Params}(f,M)$ and the likelihood $\Pr(M|f,e_0^*,e_1^*)$. %Lemma \ref{lemma:likelihood} guarantees that this will give a global maximum likelihood estimate for $f,e_0,e_1$.
 The number of all possible mappings is, however, exponential in the number of items. Given $n=|\mathbf{I}|$ items, we can assign a value of either $0$ or $1$ to any of them, giving rise to a total of $2^n$ different mappings. This makes the naive algorithm prohibitively expensive.

Our algorithm is essentially a pruning based method that uses two simple insights (described below) to narrow the search for the maximum likelihood mapping. Starting with the entire set of $2^n$ possible mappings, we eliminate all those that do not satisfy one of our two requirements, and reduce the space of mappings to be considered to $O(m)$, where $m$ is the number of worker responses per item. We then show that just an exhaustive evaluation on this small set of remaining mappings is still sufficient to find a global maximum likelihood mapping.

%The number of possible mappings $f:\mathbf{I}\rightarrow \{0,1\}$ is $2^{n}$ where $n=|\mathbf{I}|$. %In our Example \ref{ex:filter}, for instance, each item can be mapped to two values (0/1), resulting in $2^4$ possible mappings.
%In most applications, the set of items, $\mathbf{I}$ is typically very large making an extensive evaluation of $\Pr(M|f)$ for all mappings $f$ prohibitive. This raises the need for a more efficient algorithm.
%We apply our our bucketizing and dominance intuition described in Section \ref{overview} to reduce the space of considered mappings to $O(m)$, and only compute $\Pr(M|f)$ over the mappings in this smaller space. We illustrate our idea of dominance-consistent mappings on our example from Section \ref{formal}. %\ref{ex:filter}.
We illustrate our ideas on the example from Section \ref{formal}, represented graphically in Figure \ref{fig:filter-example}. We will explain this figure below.

\begin{figure}
\centering
\includegraphics[scale=0.4]{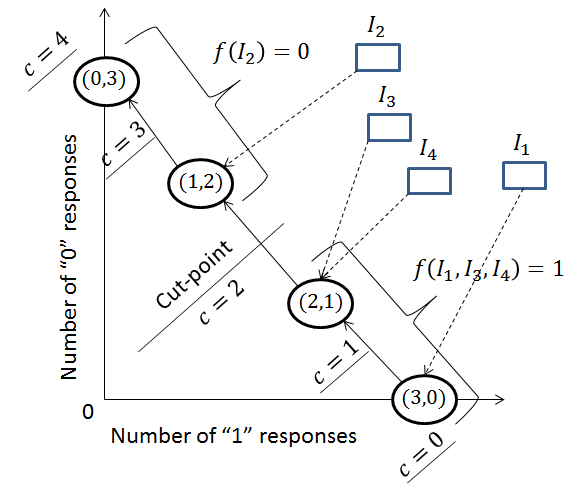}
\vspace{-10pt}
\caption{Filtering: Example \ref{ex:filter} (cut-point = 2)}
\vspace{-15pt}
\label{fig:filter-example}
\end{figure}
 %\ref{ex:filter}.
%Let us apply our bucketizing and dominance intuition described in Section \ref{overview} to the Example \ref{ex:filter}.

\smallskip
\noindent {\bf Bucketizing.}
Since we assume (for now) that all workers draw their responses from the same probability matrix $p$ (i.e., have the same $e_0,e_1$ values), we observe that items with the exact same set of worker responses can be treated identically.
This allows us to bucket items based on their observed response sets. Given that there are $m$ worker responses for each item, we have $m+1$ buckets, starting from $m$ ``1'' and zero ``0'' responses, down to zero ``1'' and $m$ ``0'' responses. %We represent these buckets in Figure \ref{fig:bucketizing}.
%{\centering
%\begin{figure}
%\includegraphics[scale=0.5]{filter-bucket.png}
%\caption{Bucketizing Items}
%\label{fig:bucketizing}
%\end{figure}
%}
We represent these buckets in Figure~\ref{fig:filter-example}. The x-axis represents the number of $1$ responses an item receives and the y-axis represents the number of $0$ responses an item receives. Since every item receives exactly $m$ responses, all possible response sets lie along the line $x+y=m$. We hash items into the buckets corresponding to their observed response sets.
Intuitively, since all items within a bucket receive the same set of responses and are for all purposes identical, two items within a bucket should receive the same value. It is more reasonable to give both items a value of 1 or 0 than to give one of them a value of 1 and the other 0. %We prune the set of mappings by only considering those mappings that give the same value to all items in a common bucket.

In our example (Figure \ref{fig:filter-example}), the set of possible responses to any item is $\{(3,0),(2,1),(1,2),(0,3)\}$, where $(3-j,j)$ represents seeing $3-j$ responses of ``1'' and $j$ responses of ``0''. %First, using our bucketizing idea, we discard all mappings which assign different values to items with the same set of responses.
We have $I_1$ in the bucket $(3,0)$, $I_3,I_4$ in the bucket (2,1), $I_2$ in the bucket $(1,2)$, and an empty bucket $(0,3)$. We only consider mappings, $f$, where items in the same bucket are assigned the same value, that is, $f(I_3)=f(I_4)$. This leaves $2^4$  mappings corresponding to assigning a value of 0/1 to each bucket. In general, given $m$ worker responses per item, we have $m+1$ buckets and $2^{m+1}$ mappings that satisfy our bucketizing condition. Although for this example $m+1=n$, typically we have $m\ll n$. %Recall that in our Barack Obama image example in Section \ref{intro}, $|\mathbf{I}|$ is the set of all images considered, which could be the set of all images obtained from a web crawl.

\smallskip
\noindent {\bf Dominance Ordering.} Second, we observe that buckets have an inherent ordering. If workers are better than random, that is, if their false positive and false negative error rates are less than 0.5, we intuitively expect items with more ``1'' responses to be more likely to have true value $1$ than items with fewer ``1'' responses. Ordering buckets by the number of ``1'' responses, we have $(m,0)\rightarrow(m-1,1)\rightarrow\ldots\rightarrow(1,m-1)\rightarrow(0,m)$, where bucket $(m-j,j)$ contains all items that received $m-j$ ``1'' responses and $j$ ``0'' responses. We eliminate all mappings that give a value of $0$ to a bucket with a larger number of ``1'' responses while assigning a value of $1$ to a bucket with fewer ``1'' responses. We formalize this intuition as a {\em dominance relation}, or ordering on buckets, $(m,0)>(m-1,1)>\ldots>(1,m-1)>(0,m)$, and only consider mappings where dominating buckets receive a value not lower than any of their dominated buckets. %We represent this dominance ordering graphically in Figure \ref{fig:filter-dominance}.
%{\centering
%\begin{figure}
%\includegraphics[scale=0.5]{filter-dominance.png}
%\caption{Filtering: Dominance Ordering}
%\label{fig:filter-dominance}
%\end{figure}
%}
%

Let us impose this dominance ordering on our example. For instance, $I_1$ (three workers respond ``1'') is more likely to have ground truth value ``1'', or dominates, $I_3,I_4$, (two workers respond ``1''), which in turn dominate $I_2$. So, we do not consider mappings that assign a value of ``0'' to a $I_1$ and ``1'' to either of $I_3,I_4$. Figure~\ref{fig:filter-example} shows the dominance relation in the form of directed edges, with the source node being the dominating bucket and the target node being the dominated one. Combining this with our bucketizing idea, we discard all mappings which assign a value of ``0'' to a dominating bucket (say response set $(3,0)$) while assigning a value of ``1'' to one of its dominated buckets (say response set $(2,1)$).

\smallskip
\noindent {\bf Dominance-Consistent Mappings.} We consider the space of mappings satisfying our above bucketizing and dominance constraints, and call them \emph{dominance}\emph{-consistent mappings}. We can prove that the maximum likelihood mapping from this small set of mappings is in fact a global maximum likelihood mapping across the space of all possible ``reasonable'' mappings: mappings corresponding to better than random worker behavior.

To construct a mapping satisfying our above two constraints, we choose a \emph{cut-point} to cut the ordered set of response sets into two partitions. The corresponding \emph{dominance-consistent mapping} then assigns value ``1'' to all (items in) buckets in the first (better) half, and value ``0'' to the rest. For instance, choosing the cut-point between response sets $(2,1)$ and $(1,2)$ in Figure \ref{fig:filter-example}  results in the corresponding dominance-consistent mapping, where $\{I_1,I_3,I_4\}$, are mapped to ``1'', while $\{I_2\}$, is mapped to ``0''. %Including the two trivial dominance-consistent mappings which either assign all buckets to ``0'', or assign all to ``1'',
We have 5 different cut-points, $0,1,2,3,4$, each corresponding to one dominance-consistent mapping. Cut-point 0 corresponds to the mapping where all items are assigned a value of 0 and cut-point 4 corresponds to the mapping where all items are assigned a value of 1. In particular, the figure shows the dominance-consistent mapping corresponding to the cut-point $c=2$. In general, if we have $m$ responses to each item, we obtain $m+2$ dominance-consistent mappings.

\begin{definition}[Dominance-consistent mapping $f^c$]
\label{cut_function}
For any cut-point $c\in 0,1,\ldots,m+1$, we define the corresponding dominance-point mapping $f^c$ as
\begin{equation*}
f^c(I)=
\begin{cases}
  1  &  \text{if } M(I)\in\{(m,0),\ldots,(m-c+1,c-1)\} \\
  0  &  \text{if } M(I)\in\{(m-c,c),\ldots,(0,m)\} \\
\end{cases}
\end{equation*}
\end{definition}

%Figure \ref{fig:filter-example} shows the above example graphically. $I_1$ belongs to the bucket corresponding to response set $(3,0)$, $I_3,I_4$ belong to the bucket with response set $(2,1)$ and $I_2$ belongs to the bucket with response set $(1,2)$. There are 5 possible cut-points, $0,1,2,3,4$, each corresponding to one dominance-consistent mapping. We show the dominance-consistent mapping corresponding to the cut-point $c=2$.

Our algorithm enumerates all dominance-consistent mappings, computes their likelihoods, and returns the most likely one among them. \papertext{The details of our algorithm can be found in our technical report.} As there are $(m+2)$ mappings, each of whose likelihoods can be evaluated in $O(m|\mathbf{I}|)$, (See Section \ref{formal}) the running time of our algorithm is $O(m^2|\mathbf{I}|)$.

\techreport{
\begin{algorithm}%[h!]
\caption{Cut-point Algorithm}
\label{binary_sampling_algo}
\begin{algorithmic}[1]
\STATE $I:=\text{Input Item-set}$
\STATE $M:=\text{Input Response Set}$
\STATE $f:=\{\}$ \COMMENT{Different dominance-consistent (mappings)}
%\STATE $s:=\{\}$ \COMMENT{Selectivities corresponding to cut-functions}
\STATE $e_0:=\{\}$ \COMMENT{$e_0$ rates corresponding to cut-functions}
\STATE $e_1:=\{\}$ \COMMENT{$e_1$ rates corresponding to cut-functions}
\STATE $Likelihood:=\{\}$ \COMMENT{Likelihoods corresponding to cut-functions}
\\ \COMMENT{Enumerating across cut-points:}
\FOR{$c$ in $\{0,1,\ldots,m+1\}$}
    \STATE{$e_0[c],e_1[c]:=\text{Params}(f[c],M)$}
    \STATE{$Likelihood[c]:=\Pr(M|f[c],e_0[c],e_1[c])$}
\ENDFOR
\STATE $c^*:=\underset{c}{\arg\!\max} Likelihood[c]$
\STATE RETURN{$(f[c^*],e_0[c^*],e_1[c^*])$}
\end{algorithmic}
\end{algorithm}
}

In the next section we prove that in spite of only searching the much smaller space of dominance-consistent mappings, our algorithm finds a global maximum likelihood solution.

\subsection{Proof of Correctness}
\label{sec:filter_proof}
\smallskip
\noindent {\bf Reasonable Mappings.}
A reasonable mapping is one which corresponds to a better than random worker behavior. Consider a mapping $f$ corresponding to a false positive rate of $e_0> 0.5$ (as given by $\text{Params}(f,M)$). This mapping is \emph{unreasonable} because workers perform worse than random for items with value ``0''. Given $\mathbf{I},M$, let $f:\mathbf{I}\rightarrow\{0,1\}$, with $e_0,e_1=\text{Params}(f,M)$ be a mapping such that $e_0\leq 0.5$ and $e_1\leq 0.5$. Then $f$ is a \emph{reasonable mapping}. It is easy to show that all dominance-consistent mappings are reasonable mappings.

Now, we present our main result on the optimality of our algorithm. We show that in spite of only considering the local space of dominance-consistent mappings, we are able to find a global maximum likelihood mapping.
\begin{theorem}[Maximum Likelihood] We let $M$ be the given response set on the input item-set $\mathbf{I}$. Let $\mathbf{F}$ be the set of all reasonable mappings and $\mathbf{F}^{dom}$ be the set of all dominance-consistent mappings. %We claim that the set of cut-functions or cut-mappings (as defined in Definition \ref{cut_function}) contains at least one mapping with optimal or maximum likelihood. That is,
Then,
\[
\max_{f^*\in \mathbf{F^{dom}}}\Pr(M|f^*)=\max_{f\in \mathbf{F}}\Pr(M|f)
\]
%\[
%\max_{c,f^c\in \mathbf{F}}(\prod_{i\in \mathbf{I}} \Pr(f^c(i)=T(i)))=\max_{f\in F}(\prod_{i\in I} \Pr(f(i)=T(i)))
%\]
\end{theorem}
\begin{proof} %We provide \papertext{an outline of} our proof below. \papertext{The complete details can be found in our technical report \cite{optTR}.} %(---cite here---).
We divide our proof into steps. The first step describes the overall flow of the proof and provides the high level structure for the remaining steps.
Step 1: Suppose $f$ is not a dominance-consistent mapping. Then, either it does not satisfy the bucketizing constraint, or it does not satisfy the dominance-constraint. We claim that if $f$ is reasonable, we can always construct a dominance-consistent mapping, $f^*$ such that $\Pr(M|f^*)\geq \Pr(M|f)$. Then, it trivially follows that $\max_{f^*\in \mathbf{F^{dom}}}\Pr(M|f^*)=\max_{f\in \mathbf{F}}\Pr(M|f)$. We show the construction of such an $f^*$ for every reasonable mapping $f$ in the following steps.

Step 2: (Dominance Inconsistency). Suppose $f$ does not satisfy the dominance constraint. Then, there exists at least one pair of items $I_1,I_2$ such that $M(I_1)>M(I_2)$ and $f(I_1)=0,f(I_2)=1$. Define mapping $f'$ as follows:
\begin{equation*}
f'(I)=
\begin{cases}
  f(I_2)  &  \text{for } I=I_1 \\
  f(I_1)  &  \text{for } I=I_2 \\
  f(I) & \forall I\in \mathbf{I}\setminus\{I_1,I_2\}\\
\end{cases}
\end{equation*}
Mapping $f'$ is identical to $f$ everywhere except for at $I_1$ and $I_2$, where it swaps their respective values. We show that $\Pr(M|f')\geq \Pr(M|f)$. Let $M(I_1)=(m-i,i)$ and $M(I_2)=(m-j,j)$ where $i<j$ and $f(I_1)=0,f(I_2)=1$. Let $n_{k1}$ (respectively $n_{k0}$) denote the number of items, {\em excluding $I_1,I_2$}, with response set $(m-k,k)$ in $M$ such that $f(I)=1$ (respectively 0). We abuse notation slightly to use $n_{k1},n_{k0}$ to also denote the sets of items, {\em excluding $I_1,I_2$}, with response set $(m-k,k)$ and a value of $1,0$ respectively under $f$, wherever the meaning is clear from the context. Given $f,M$, we can calculate the response probability matrix as shown in Section \ref{sec:filter_correspondence}. Let $p_{00}=1-e_0$ be the probability that workers respond 0 to an item with mapping $0$ under $f$. Similarly, $p_{11}=1-e_1$ be the probability that workers respond 1 to an item with mapping $1$ under $f$. Given $f,e_0,e_1$, we have $\Pr(M|f)=\underset{I}\prod \Pr(M(I)|f,e_0,e_1)$.
We can write
\[
p_{11} = \frac{(m-j)+\underset{k}\sum (m-k)n_{k1}}{m+\underset{k}\sum mn_{k1}}
\]
and
\[
p_{00} = \frac{i+\underset{k}\sum (m-k)n_{k0}}{m+\underset{k}\sum mn_{k0}}
\]
We split the likelihood of $f$ into two independent parts. Let the probability contributed by items in $n_{k1}$ be $\Pr_1(M|f)$, and that contributed by $n_{k0}$ be $\Pr_0 (M|f)$, such that $\Pr(M|f)=\Pr_1(M|f)\Pr_0(M|f)$. We claim that for reasonable mappings, $\Pr_1(M|f')\geq \Pr_1(M|f)\land \Pr_0(M|f')\geq \Pr_0(M|f)$. Then, we have $\Pr(M|f')\geq \Pr(M|f)$. We prove below that $\Pr_1(M|f')\geq \Pr_1(M|f)$ The proof for $\Pr_0$ can be derived in a similar fashion.
We have
$$
\Pr_1(M|f)=p_{11}^{(m-j)+\underset{k}\sum (m-k)n_{k1}}(1-p_{11})^{j+\underset{k}\sum kn_{k1}}
$$
We can similarly calculate $\Pr_1(M|f')$, where $p_{11}'=\frac{(m-i)+\underset{k}\sum (m-k)n_{k1}}{m+\underset{k}\sum mn_{k1}}$. Now, let $a=(m-i)+\underset{k}\sum (m-k)n_{k1}$, $b=(m-j)+\underset{k}\sum (m-k)n_{k1}$, and $c=m+\underset{k}\sum mn_{k1}$. We then have $\frac{\Pr_1(M|f')}{\Pr_1(M|f)}=\frac{a^a(c-a)^{c-a}}{b^b(c-b)^{c-b}}$. Note that since $i<j$, we have $a>b$. It can then be shown that for $a+b\geq c$, $\frac{\Pr_1(M|f')}{\Pr_1(M|f)}\geq 1$. Furthermore, for reasonable mappings, we have $p_{11}\geq\frac{1}{2}\Rightarrow a,b\geq \frac{c}{2}\Rightarrow a+b\geq c$. Therefore, $\frac{\Pr_1(M|f')}{\Pr_1(M|f)}\geq 1\Rightarrow \Pr_1(M|f')\geq\Pr_1(M|f)$. Similarly, we can show $\Pr_0(M|f')\geq\Pr_0(M|f)$, and therefore, $\Pr(M|f')\geq\Pr(M|f)$.

Step 3: (Bucketizing Inconsistency). Suppose $f$ does not satisfy the bucketizing constraint. Then, we have at least one pair of items $I_1,I_2$ such that $M(I_1)=M(I_2)$ and $f(I_1)\neq f(I_2)$. Consider the two mappings $f_1$ and $f_2$ defined as follows:
\begin{equation*}
f_1(I)=
\begin{cases}
  f(I_2)  &  \text{for } I=I_1 \\
  f(I) & \forall I\in \mathbf{I}\setminus\{I_1\}\\
\end{cases}
\end{equation*}
\begin{equation*}
f_2(I)=
\begin{cases}
  f(I_1)  &  \text{for } I=I_2 \\
  f(I) & \forall I\in \mathbf{I}\setminus\{I_2\}\\
\end{cases}
\end{equation*}
The mappings $f_1$ and $f_2$ are identical to $f$ everywhere except for at $I_1$ and $I_2$, where $f_1(I_1)=f_1(I_2)=f(I_2)$ and $f_2(I_1)=f_2(I_2)=f(I_1)$. We \techreport{can} show \techreport{(using a similar calculation as in Step 2)} that $\max (\Pr(M|f_1),\Pr(M|f_2))\geq \Pr(M|f)$. Let $f'=\arg\!\max_{f_1,f_2}(\Pr(M|f_1),\Pr(M|f_2))$.

Step 4: (Reducing Inconsistencies). Suppose $f$ is not a dominance-consistent mapping. We have shown that by reducing either a bucketizing inconsistency (Step 2), or a dominance inconsistency (Step 3), we can construct a new mapping, $f'$ with likelihood greater than or equal to that of $f$. Now, if $f'$ is a dominance-consistent mapping, set $f^*=f'$ and we are done. If not, look at an inconsistency in $f'$ and apply steps 2 or 3 to it. \papertext{We can show that with} \techreport{With} each iteration, we are reducing at least one inconsistency while increasing likelihood. We repeat this process iteratively, and since there are only a finite number of inconsistencies in $f$ to begin with, we are guaranteed to end up with a desired dominance-consistent mapping $f^*$ satisfying $\Pr(M|f^*)\geq \Pr(M|f)$. This completes our proof\papertext{ sketch}. \qed
\end{proof}

%\subsection{Correspondence between mappings and worker error rates}
\subsection{Calculating error rates from mappings}
\label{sec:filter_correspondence}
%\agp{Fixed - consider moving everything below to a section 3.5; this is distracting
%--- it is easy to appreciate the fact that given a mapping,
%finding error rates is straightforward. Maybe say something like. For now,
%we focus on the problem of finding the maximum likelihood mapping
%with the understanding that finding the error rates is
%straightforward. BTW, does this also hold for rating? If so, we should
%bring it up in preliminaries....}

In this section, we formalize the correspondence between mappings and worker error rates that we introduced in Section \ref{formal}. Given a response set $M$ and a mapping $f$, say we calculate the corresponding worker error rates $e_0(f,M),e_1(f,M)$ $=$ $\text{Params}(f,M)$
%$s(f,M),e_0(f,M),e_1(f,M)=\text{Params}(f,M)$ as:
as follows:
\begin{enumerate}
%\item Let $\mathbf{I'}\subseteq\mathbf{I}\ni f(I)=1\forall I\in\mathbf{I'}$ and $f(I)=0\forall I\in\mathbf{I\setminus I'}$. Then, $s=\frac{|\mathbf{I'}|}{|\mathbf{I}|}$.
\item Let $\mathbf{I_j}\subseteq\mathbf{I}\ni f(I)=0,M(I)=(m-j,j)\forall I\in\mathbf{I_j}$. Then, $e_0=\frac{\sum_{j=0}^m (m-j)|I_j|}{m\sum_{j=0}^m |I_j|}$
\item Let $\mathbf{I_j}\subseteq\mathbf{I}\ni f(I)=1,M(I)=(m-j,j)\forall I\in\mathbf{I_j}$. Then, $e_1=\frac{\sum_{j=0}^mj|I_j|}{m\sum_{j=0}^m |I_j|}$
\end{enumerate}
Intuitively, $e_0(f,M)$ (respectively $e_1(f,M)$) is just the fraction of times a worker responds with a value of 1 (respectively 0) for an item whose true value is 0 (respectively 1), under response set $M$ assuming that $f$ is the ground truth mapping.
%We write $e_0(f,M),e_1(f,M)=\text{Params}(f,M)$.
We show that for each mapping $f$, this intuitive set of false positive and false negative error rates, $(e_0,e_1)$ maximizes $\Pr(M|f,e_0,e_1)$. We express this idea formally below.

%Note that $\text{Params}(f,M)$ is not the only set of values that can result in the observed response set $M$.
%Note that due to the probabilistic nature of our problem, different response probability matrices can produce the same $M$ under $f$, but with varying likelihoods of occurrence. $\text{Params}(f,M)$ is the most likely such set of parameters. Lemma \ref{params} states this formally. Its proof can be found in the full technical report (---cite here---).
%\agp{fixed - what does this lemma talk about.. say something before you dive in}
\begin{lemma}[$\text{Params}(f,M)$]
\label{params}
%Let mapping $f$ be the given ground truth on item set $\mathbf{I}$ over evidence $M$. Then, %$\text{Params}(f,M)=\argmax_{s,e_0,e_1}\Pr(s,e_0,e_1|f,M)$
Given response set $M$. Let $\Pr(e_0,e_1|f,M)$ be the probability that the underlying worker false positive and negative rates are $e_0$ and $e_1$ respectively, conditioned on mapping $f$ being the true mapping. Then, \[\forall f, \text{Params}(f,M)=\underset{e_0,e_1}{\arg\!\max} \Pr(e_0,e_1|f,M)\]
\end{lemma}
\begin{proof}
Let $M,f$ be given. By Bayes theorem, \\$\Pr(e_0,e_1|f,M)=k\Pr(M|e_0,e_1,f)$ for some constant $k$. Therefore, $\underset{e_0,e_1}{\arg\!\max} \Pr(e_0,e_1|f,M)=\underset{e_0,e_1}{\arg\!\max} \Pr(M|e_0,e_1,f)$.
Now, $\Pr(M|e_0,e_1,f)=\underset{I\in\mathbf{I}}\prod \Pr(M(I)|f,e_0,e_1)$.

Let $\mathbf{I_{j,0}}\subseteq\mathbf{I}\ni f(I)=0,M(I)=(m-j,j)\forall I\in\mathbf{I_j}$ and $\mathbf{I_{j,1}}\subseteq\mathbf{I}\ni f(I)=1,M(I)=(m-j,j)\forall I\in\mathbf{I_j}$. We have
\[
\Pr(M(I)|f,e_0,e_1)=
\begin{cases}
(1-e_1)^{m-j}e_1^j \forall I\in\mathbf{I_{j,1}}\\
e_0^{m-j}(1-e_0)^j \forall I\in\mathbf{I_{j,0}}\\
\end{cases}
\]
Therefore, $\Pr(M|f,e_0,e_1)=\underset{j}\prod[(1-e_1)^{m-j}e_1^j]^{|\mathbf{I_{j,1}}|}[e_0^{m-j}(1-e_0)^j]^{|\mathbf{I_{j,0}}|}$. For ease of notation, let %$(1-e_1)^{m-j}e_1^j=x_j$, $e_0^{m-j}(1-e_0)^j=y_j$,
$a_1=\overset{m}{\underset{j=0}\sum }j|I_{j,1}|$, $b_1=\overset{m}{\underset{j=0}\sum }(m-j)|I_{j,1}|$, $a_0=\overset{m}{\underset{j=0}\sum }j|I_{j,0}|$, and $b_0=\overset{m}{\underset{j=0}\sum }(m-j)|I_{j,0}|$. Then, we have $\Pr(M|f,e_0,e_1)=(1-e_1)^{b_1}e_1^{a_1}(1-e_0)^{a_0}e_0^{b_0}$.

To maximize $\Pr(M|e_0,e_1,f)$ given $M,f$, we compute its partial derivates with respect to $e_0$ and $e_1$ and set them to 0. We have, $\frac{{\partial \Pr(M|e_0,e_1,f)}}{\partial e_1}=0\Rightarrow a_1(1-e_1)-b_1e_1=0$ (simplifying common terms). Therefore, $e_1=\frac{a_1}{a_1+b_1}=\frac{\sum_{j=0}^mj|I_j|}{m\sum_{j=0}^m |I_j|}$.  It is easy to verify that the second derivative $\frac{\partial^2 \Pr(M|e_0,e_1,f)}{\partial e_1^2}$ is negative for this value of $e_1$. We recall that this is the value of $e_1$ under $\text{Params}(f,M)$. Similarly, we can also show that $\frac{{\partial \Pr(M|e_0,e_1,f)}}{\partial e_0}=0$ forces $e_0=\frac{\sum_{j=0}^m (m-j)|I_j|}{m\sum_{j=0}^m |I_j|}$, which is the value given by $\text{Params}(f,M)$.

Therefore, $\text{Params}(f,M)=\underset{e_0,e_1}{\arg\!\max} \Pr(e_0,e_1|f,M)$.
\qed \end{proof}

%\smallskip
%\noindent {\bf Maximum Likelihood Mapping.}
%We leverage this correlation between mappings and response probability matrices to simplify our problem.
\noindent Next, we show that instead of simultaneously trying to find the most likely mapping and false positive and false negative error rates, it is sufficient to just find the most likely mapping while assuming that the error rates corresponding to any chosen mapping, $f$, are always given by $\text{Params}(f,M)$. We formalize this intuition in Lemma \ref{lemma:likelihood} below. \papertext{The details of its proof, which follows from Lemma \ref{params}, can also be found in the full technical report \cite{optTR}.}% (---cite here---)\agp{XXX}.

\begin{lemma}[Likelihood of a Mapping]
\label{lemma:likelihood}
Let $f\in\mathbf{F}$ be any mapping and $M$ be the given response set on $\mathbf{I}$. We have, \[\underset{f,e_0,e_1}{\max} \Pr(M|f,e_0,e_1)=\underset{f}{\max} \Pr(M|f,\text{Params}(f,M))\] %Also, we can evaluate $\Pr(M|f,\text{Params}(f,M))$ in $O(m|\mathbf{I}|)$.
\end{lemma}
%\noindent{\bf Proof.} Follows from Lemma \ref{params}.
%\qedr \\ \linebreak
\begin{proof}
The proof for this statement follows easily from Lemma \ref{params}. Let $f^*,e_0^*,e_1^*=\underset{f,e_0,e_1}{\max} \Pr(M|f,e_0,e_1)$. Now, let $e_0',e_1'=\text{Params}(f^*,M)$. From Lemma \ref{params}, we have $e_0',e_1'=\underset{e_0,e_1}{\arg\!\max} \Pr(M|f^*,e_0,e_1)$. So we have, $\Pr(M|f^*,\text{Params}(f^*,M))\geq \Pr(M|f^*,e_0^*,e_1^*)$. Additionally, $\underset{f}{\max} \Pr(M|f,\text{Params}(f,M))\geq \Pr(M|f^*,\text{Params}(f^*,M))$. Therefore, \[\underset{f}{\max} \Pr(M|f,\text{Params}(f,M))\geq \underset{f,e_0,e_1}{\max} \Pr(M|f,e_0,e_1)\]

But $\underset{f,e_0,e_1}{\max} \Pr(M|f,e_0,e_1)\geq \underset{f}{\max} \Pr(M|f,\text{Params}(f,M))$ by definition.

Therefore, combining the above inequalities, we have \[\underset{f,e_0,e_1}{\max} \Pr(M|f,e_0,e_1)=\underset{f}{\max} \Pr(M|f,\text{Params}(f,M))\qed\]
%\qed
\end{proof}

\subsection{Experiments}
\label{sec:filter_exp}
%\agp{fixed -- this section should be in present tense.}
The goal of our experiments is two-fold. First, we wish to verify that our algorithm does indeed find higher likelihood mappings. Second, we wish to compare our algorithm against standard baselines for such problems, like the EM algorithm, for different metrics of interest. While our algorithm optimizes for likelihood of mappings, we are also interested in other metrics that measure the quality of our predicted item assignments and worker response probability matrix. For instance, we test what fraction of item values are predicted correctly by different algorithms to measure the quality of item value prediction. We also compare the similarity of predicted worker response probability matrices with the actual underlying matrices using distance measure like Earth-Movers Distance and Jensen-Shannon Distance. We run experiments on both simulated as well as real data and discuss our findings below.

%\smallskip
%\noindent {\bf Simulated Data.}
\subsubsection{Simulated Data}
\label{sec:filter-exp-simulated}
%\agp{fixed -- organize these under dataset generation, parameters varied,
%metrics subheadings. Also, when you mention the figures, also provide their
%description there itself instead of doing it later. }

\smallskip
\noindent {\bf Dataset generation.}
For our synthetic experiments, we assign ground truth 0-1 values to $n$ items randomly based on a fixed selectivity. Here, a selectivity of $s$ means that each item has a probability of $s$ of being assigned true value $1$ and $1-s$ of being assigned true value $0$. This represents our set of items $\mathbf{I}$ and their ground truth mapping, $T$. We generate a random ``true'' or underlying worker response probability matrix, with the only constraint being that workers are better than random (false positive and false negative rates $\leq 0.5$). We simulate the process of workers responding to items by drawing their response from their true response probability matrix, $p_{\text{true}}$. This generates one instance of the response set, $M$. Different algorithms being compared now take $\mathbf{I},M$ as input and return a mapping $f:\mathbf{I}\rightarrow \{0,1\}$, and a worker response matrix $p$.

\smallskip
\noindent {\bf Parameters varied.}
We experiment with different choices over both input parameters and comparison metrics over the output.  For input parameters, we vary the number of items $n$, the selectivity (which controls the ground truth) $s$, and the number of worker responses per item, $m$. While we try out several different combinations, we only show a small representative set of results below. In particular, we observe that changing the value of $n$ does not significantly affect our results. We also note that the selectivity can be broadly classified into two categories: evenly distributed ($s=0.5$), and skewed ($s>0.5$ or $s<0.5$). In each of the following plots we use a set of $n=1000$ items and show results for either $s=0.5$ or $s=0.7$. We show only one plot if the result is similar to and representative of other input parameters. More \papertext{extensive }experimental results can be found in \papertext{our full technical report \cite{optTR}.} \techreport{the appendix, Section \ref{sec:appendix-filter-exp}.} %(---cite here---)\agp{cite here}.

%We randomly generate worker response probability matrices (false positive and false negative error rates) and vary the number of worker responses per item.
\smallskip
\noindent {\bf Metrics.} We test the output of different algorithms on a few different metrics: we compare the likelihoods of their output mappings, we compare the fraction of items whose values are predicted incorrectly, and we compare the quality of predicted worker response probability matrix. For this last metric, we use different distance functions to measure how close the predicted worker matrix is to the underlying one used to generate the data. In this paper we report our distance measures using an Earth-Movers Distance (EMD) based score~\cite{emd}. For a full description of our EMD based score and other distance metrics used, we refer to \papertext{our full technical report \cite{optTR}.} \techreport{the appendix, Section \ref{sec:appendix-filter-exp}.} %(---cite here---)\agp{cite here}.

%EMD is a metric function that captures how similar two probability distributions are. Intuitively, if the two distributions are represented as piles of sand, EMD is a measure of the minimum amount of sand that needs to be shifted to make the two piles equal. In our problem, the worker response matrix $p$ can be represented as two probability distribution corresponding to $p(i,1)$ and $p(i,0)$, that is the probability distributions of worker responses given that the true value of an item is 1 and 0 respectively. We compute the EMD of $p(i,1)$ from $p_{\text{true}}(i,1)$ and $p(i,0)$ from $p_{\text{true}}(i,0)$ and record their sum as the EMD ``score'' of the algorithm that predicts $p$. Since the EMD between $p(i,j)$ and $p_{\text{true}}(i,j)$ lies between $[0,1]$, our EMD score that sums the individual EMDs for $i=0,1$ lies in $[0,2]$.

\smallskip
\noindent {\bf Algorithms.} We compare our algorithm, denoted $OPT$, against the standard Expectation-Maximization (EM) algorithm that is also solving the same underlying maximum likelihood problem. The EM algorithm starts with an arbitrary initial guess for the worker response matrix, $p_1$ and computes the most likely mapping $f_1$ corresponding to it. The algorithm then in turn computes the most likely mapping $p_2$ corresponding to $f_1$ (which is not necessarily $p_1$) and repeats this process iteratively until convergence. We experiment with different initializations for the EM algorithm, represented by $EM(1),EM(2),EM(3)$.  $EM(1)$ represents the starting point with false positive and negative rates $e_0,e_1=0.25$ (workers are better than random), $EM(2)$ represents the starting point of $e_0,e_1=0.5$ (workers are random), and $EM(3)$ represents the starting point of $e_0,e_1=0.75$ (workers are worse than random). $EM(*)$ is the consolidated algorithm which runs each of the three EM instances and picks the maximum likelihood solution across them for the given $\mathbf{I},M$.

\smallskip
\noindent {\bf Setup.}
We vary the number of worker responses per item along the x-axis and plot different objective metrics (likelihood, fraction of incorrect item value predictions, accuracy of predicted worker response matrix) along the y-axis. Each data point represents the value of the objective metric averaged across 1000 random trials. That is, for each fixed value of $m$, we generate 1000 different worker response matrices, and correspondingly 1000 different response sets $M$. We run each of the algorithms over all these datasets, measure the value of their objective function and average across all problem instances to generate one point on the plot.

\smallskip
\noindent {\bf Likelihood.}
Figure \ref{fig:filter_like} shows the likelihoods of mappings returned by our algorithm $OPT$ and the different instances of the EM algorithm. In this experiment, we use $s=0.5$, that is items' true values are roughly evenly distributed over $\{0,1\}$. Note that the y-axis plots the likelihood on a log scale, and that a higher value is more desirable. We observe that our algorithm does indeed return higher likelihood mappings with the marginal improvement going down as $m$ increases. However, in practice, it is unlikely that we will ever use $m$ greater than 5 (5 answers per item). While our gains for the simple filtering setting are small, as we will see in Section~\ref{sec:rating-exp}, the gains are significantly higher for the case of rating, where multiple error rate parameters are being simultaneously estimated. (For the rating case, only false positive and false negative error rates are being estimated.)

\begin{figure*}[!t]
\vspace{-10pt}
%\centering
\subfigure{
\includegraphics[scale=0.3]{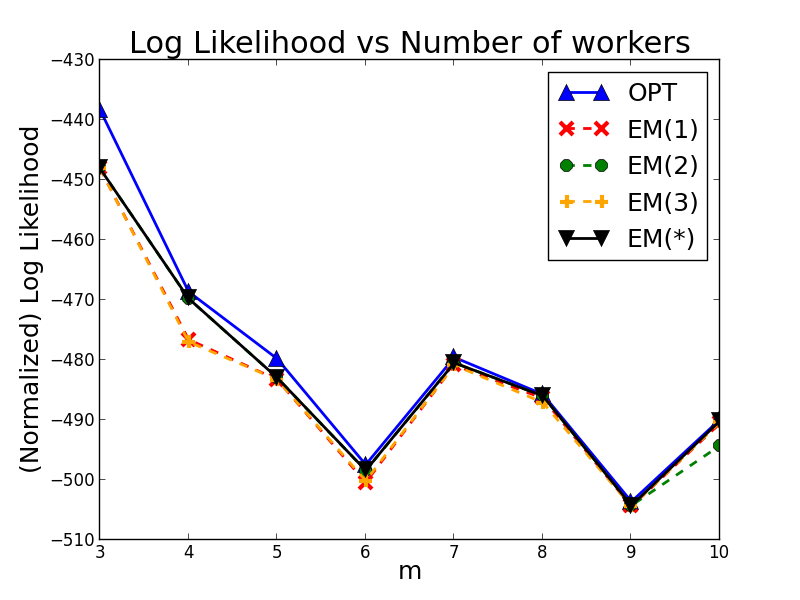}
\label{fig:filter_like}
}
\hspace{-15pt}
\subfigure{
\includegraphics[scale=0.3]{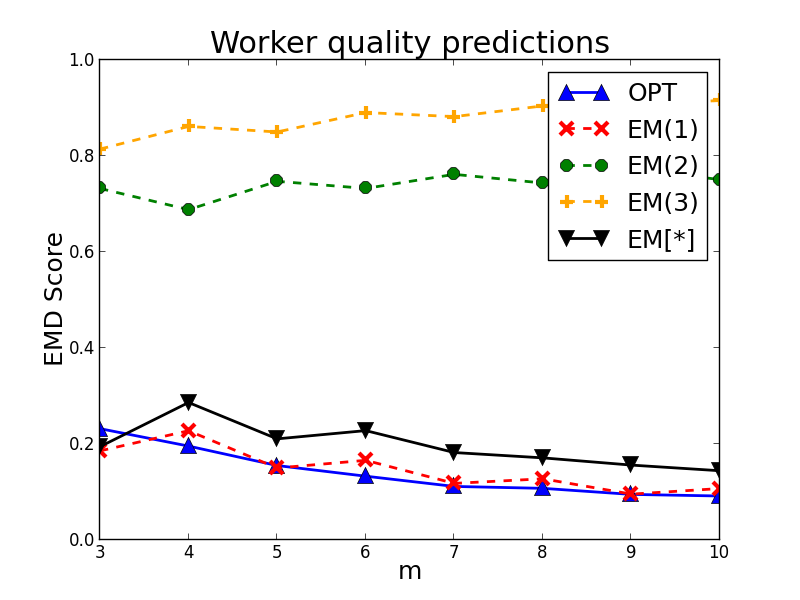}
\label{fig:filter_frac1}
}
\hspace{-15pt}
\subfigure{
\includegraphics[scale=0.3]{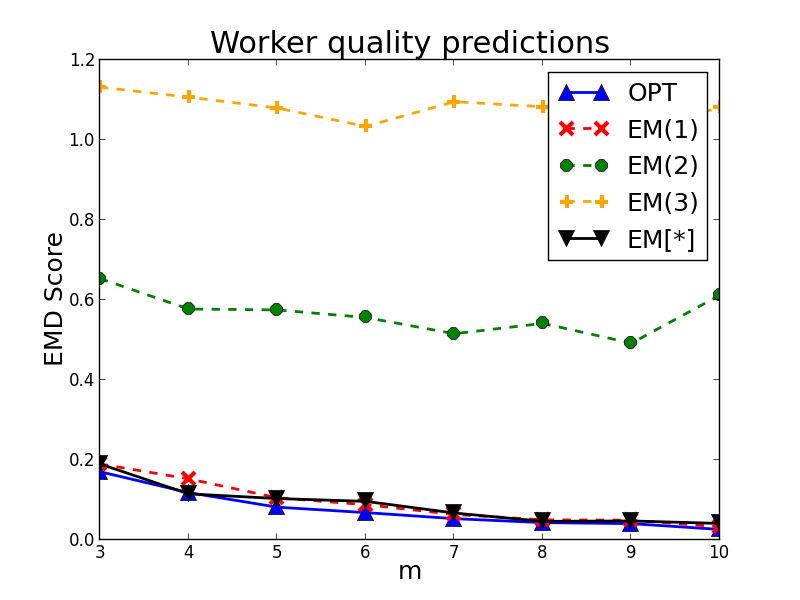}
\label{fig:filter_emd1}
}
\vspace{-20pt}
\caption{Synthetic Data Experiments: (a)Likelihood, $s=0.5$ (b) Fraction Incorrect, $s=0.7$ (c) EMD Score, $s=0.5$}
\vspace{-15pt}
\end{figure*}

% {\centering
% \begin{figure}
% %\includegraphics[scale=0.4]{like,s=5.png}
% \includegraphics[scale=0.4]{like,s=5_temp.png}
% \caption{Likelihood, $s=0.5$}
% \label{fig:filter_like}
% \end{figure}
% }

\smallskip
\noindent {\bf Fraction incorrect.}
 In Figure \ref{fig:filter_frac1} ($s=0.7$), we plot the fraction of item values each of the algorithms predicts \emph{incorrectly} and average this measure over the 1000 random instances. A lower score means a more accurate prediction. We observe that our algorithm estimates the true values of items with a higher accuracy than the EM instances.

% {\centering
% \begin{figure}
% %\includegraphics[scale=0.4]{frac_incorrect,s=7.png}
% \includegraphics[scale=0.4]{frac_incorrect,s=7_temp.png}
% \caption{Fraction Incorrect, $s=0.7$}
% \label{fig:filter_frac1}
% \end{figure}
% }

\smallskip
\noindent {\bf EMD score.}
To compare the qualities of our predicted worker false positive and false negative error rates, we compute and plot EMD-based scores in Figure \ref{fig:filter_emd1} ($s=0.5$) and Figure \ref{fig:filter_emd2} ($s=0.7$). Note that since EMD is a distance metric, a lower score means that the predicted worker response matrices are closer to the actual ones; so, algorithms that are lower on this plot do better. We observe that the worker response probability matrix predicted by our algorithm is closer to the actual probability matrix used to generate the data than all the EM instances. While $EM(1)$ in particular does well for this experiment, we  observe that $EM(2)$ and $EM(3)$ get stuck in bad local maxima making $EM(*)$ prone to the initialization.

% {\centering
% \begin{figure}
% %\includegraphics[scale=0.4]{emd,s=5.png}
% \includegraphics[scale=0.4]{emd,s=5_temp.png}
% \caption{EMD Score, $s=0.5$}
% \label{fig:filter_emd1}
% \end{figure}
% }

% {\centering
% \begin{figure}
% %\includegraphics[scale=0.4]{emd,s=7.png}
% \includegraphics[scale=0.4]{emd,s=7_temp.png}
% \caption{EMD Score, $s=0.7$}
% \label{fig:filter_emd2}
% \end{figure}
% }
% {\centering
% \begin{figure}
% %\includegraphics[scale=0.4]{ic_like_random.png}
% \includegraphics[scale=0.4]{ic_like_random_temp.png}
% \caption{Likelihood}
% \label{fig:ic_like_random}
% \end{figure}
% }
% {\centering
% \begin{figure}
% %\includegraphics[scale=0.4]{ic_frac_random.png}
% \includegraphics[scale=0.4]{ic_frac_random_temp.png}
% \caption{Fraction Incorrect}
% \label{fig:ic_frac_random}
% \end{figure}
% }

%\smallskip
%\noindent {\bf Results.}

 Although averaged across a large number of instances, $EM(1)$ and $EM(*)$ do perform well, our experiments show that optimizing for likelihood does not adversely affect other potential parameters of interest. For all metrics considered, $OPT$ performs better, in addition to giving us a global maximum likelihood guarantee. (As we will see in the rating section, our results are even better there since multiple parameters are being estimated.) We experiment with a number of different parameter settings and comparison metrics and present more extensive results in \papertext{our full technical report \cite{optTR}.} \techreport{the appendix.}

%\smallskip
%\noindent {\bf Real Data.}
\subsubsection{Real Data}
\label{sec:filter-exp-real}
\smallskip
\noindent {\bf Dataset.}
In this experiment, we use an image comparison dataset~\cite{IC} where 19 workers are asked to each evaluate 48 tasks. Each task consists of displaying a pair of sporting images to a worker and asking them to evaluate if both images show the same sportsperson. We have the ground truth yes/no answers for each pair, but do not know the worker error rates. Note that for this real dataset, our assumptions that all workers have the same error rates and answer questions independently may not necessarily hold true. We show that in spite of the assumptions made by our algorithm, they estimate the values of items with a high degree of accuracy even on this real dataset.

To evaluate the performance of our algorithm and the EM-based baseline, we compare the the estimates for the item values against the given ground truth. Note that since we do not have a ground truth for worker error rates under this setting, we cannot evaluate the algorithms for that aspect---we do however study likelihood of the final solutions
from different algorithms.

\smallskip
\noindent {\bf Setup.}
We vary the number of workers used from 1 to 19 and plot the performance of algorithms $OPT$, $EM(1)$, $EM(2)$, $EM(3)$, $EM(*)$ similar to Section \ref{sec:filter-exp-simulated}. We plot the number of worker responses used along the x-axis. For instance, a value of $m=4$ indicates that for each item, four {\em random} worker responses are chosen. The four workers answering one item may be different from those answering another item. This random sample of the response set is given as input to the different algorithms. Similar to our simulations, we average our results across 100 different trials for each data point in our subsequent plots. For each fixed value of $m$, one trial corresponds to choosing a set of $m$ worker responses to each item randomly. We run 100 trials for each $m$, and correspondingly generate 100 different response sets $M$. We run our $OPT$ and EM algorithms over all these datasets, measure the value of different objectives function and average across all problem instances to generate one point on a plot.

\smallskip
\noindent {\bf Likelihood.}
Figure \ref{fig:ic_like_random} plots the likelihoods of the final solution for different algorithms. We observe that except for $EM(2)$, all algorithms have a high likelihood. This can be explained as follows: $EM(2)$ which starts with an initialization of $e_0$ and $e_1$ rates around 0.5 and converges to a final response probability matrix in that neighborhood. Final error rates of around 0.5 (random) will have naturally low likelihood when there is a high amount of agreement between workers. $EM(1)$ and $EM(3)$ on the other hand start with, and converge to near opposite extremes with $EM(1)$ predicting $e_0 / e_1$ rates $\approx 0$ and $EM(3)$ predicting error rates $\approx 1$. Both of these, however, result in a high likelihood of observing the given response, with $EM(1)$ predicting that the worker is always correct, and $EM(3)$ predicting that the worker is always incorrect, i.e., adversarial.  Even though $EM(1)$ and $EM(3)$ often converge to completely opposite predictions of item-values because of their initializations, their solutions still have similar likelihoods corresponding to the intuitive extremes of perfect and adversarial worker behavior. This behavior thus demonstrates the strong dependence of $EM$-based approaches on the initialization parameters.

\begin{figure*}[!t]
\vspace{-15pt}
%\centering
\subfigure{
\includegraphics[scale=0.3]{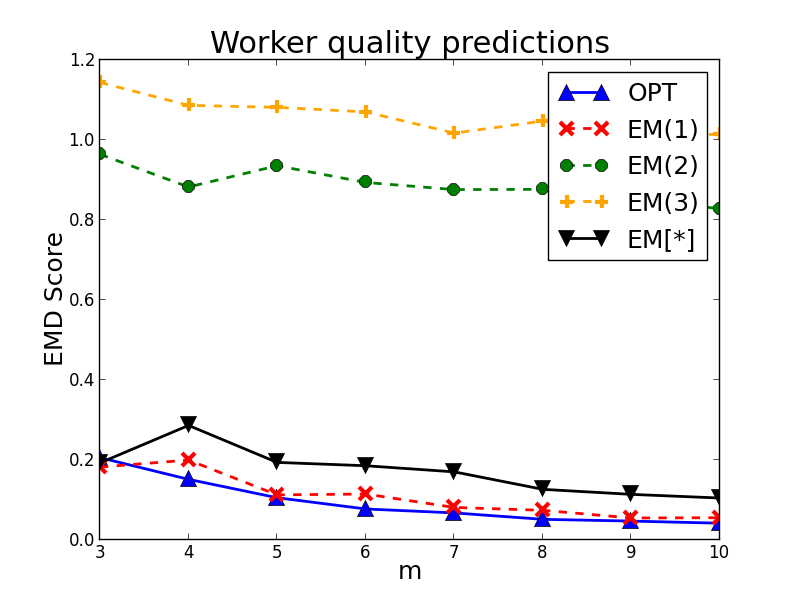}
\label{fig:filter_emd2}
}
\hspace{-15pt}
\subfigure{
\includegraphics[scale=0.3]{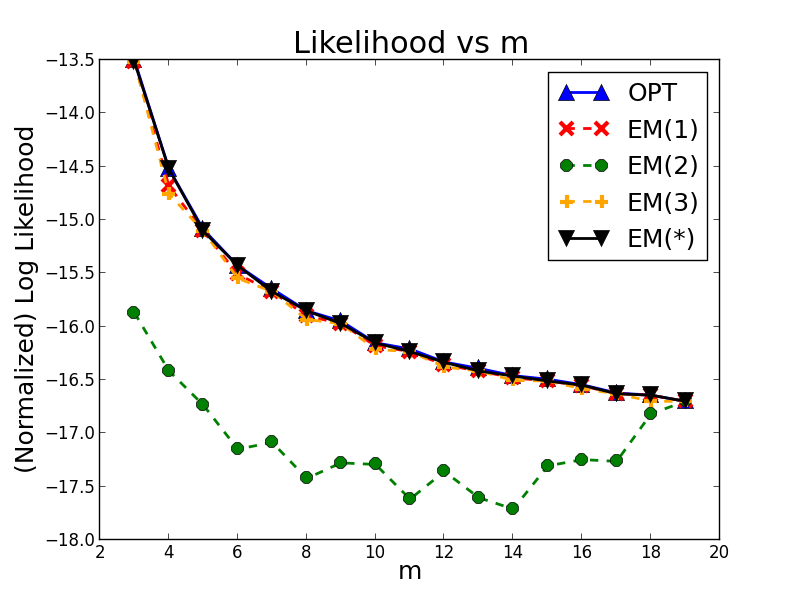}
\label{fig:ic_like_random}
}
\hspace{-15pt}
\subfigure{
\includegraphics[scale=0.3]{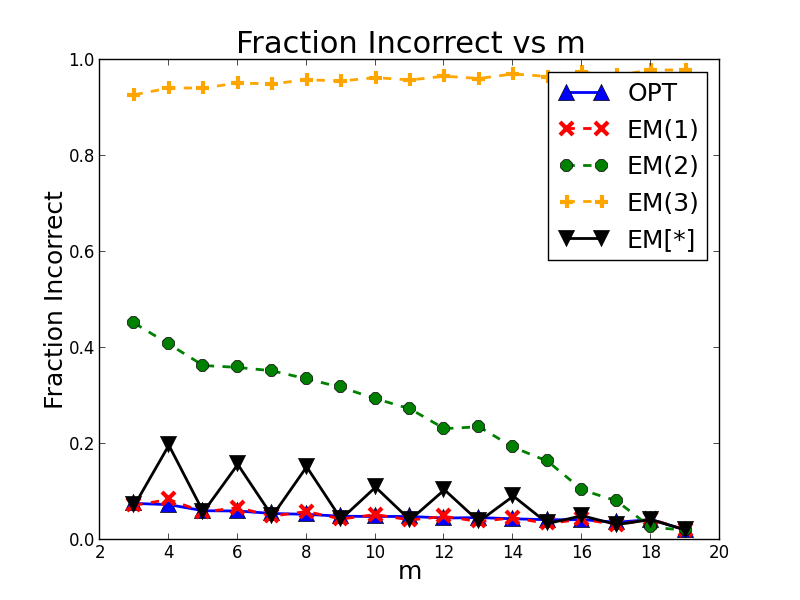}
\label{fig:ic_frac_random}
}
\vspace{-20pt}
\caption{Synthetic Data Experiments: (a) EMD Score, $s=0.7$  Real Data Experiments: (b) Likelihood (c) Fraction Incorrect}
\vspace{-15pt}
\end{figure*}

\smallskip
\noindent {\bf Fraction Incorrect.}
Figure \ref{fig:ic_frac_random} plots the fraction of items predicted incorrectly along the y-axis for $OPT$ and the EM algorithms. Correspondingly, their predictions for item values are opposite, as can be seen in Figure \ref{fig:ic_frac_random}.

We observe that both $EM(1)$ and our algorithm $OPT$ do fairly well on this dataset %with an over $90\%$ accuracy of predicting item values,
even when a very few number of worker responses are used.
However, $EM(*)$, which one may expect would typically do better than the individual $EM$ initializations,
sometimes does poorly compared to $OPT$ by picking solutions of high likelihood that are nevertheless not very good.
Note that here we assume that worker identities are unknown and arbitrary workers could be answering different tasks --- our goal is to characterize the behavior of the worker population as a whole. For larger datasets, we expect the effects of population smoothing to be greater and our assumptions on worker homogeneity to be closer to the truth. So, even though our algorithm provides theoretical global guarantees under somewhat strong assumptions, it also performs well for settings where our assumptions may not necessarily be true.

\section{Rating Problem}
In this section, we extend our techniques from filtering to the problem of rating items. Even though the main change resides in the possible values of items ($\{0,1\}$ for filtering and $\{1,\ldots,R\}$ for rating), this small change adds  significant complexity to our dominance idea. We show how our notions of bucketizing and dominance generalize from the filtering case.

\subsection{Formalization}
Recall from Section \ref{prelims} that, for the rating problem, workers are shown an item and asked to provide a score from $1$ to $R$, with $R$ being the best (highest) and $1$ being the worst (lowest) score possible. Each item from the set $\mathbf{I}$ receives $m$ worker responses and all the responses are recorded in $M$. We write $M(I)=(v_R,v_{R-1},\ldots,v_1)$ if item $I\in\mathbf{I}$ receives $v_i$ responses of ``$i$'', $1\leq i\leq R$. Recall that $\overset{R}{\underset{i=1}\sum v_i}=m$. Mappings are functions $f:\mathbf{I}\rightarrow\{1,2,\ldots,R\}$ and workers are described by the response probability matrix $p$, where $p(i,j)$ ($i,j\in \{1,2,\ldots,R\}$) denotes the probability that a worker will give an item with true value $j$ a score of $i$. Our problem is defined as that of finding $f^*,p^*=\underset{f,p}{\arg\!\max} \Pr(M|f,p)$ given $M$.

As in the case of filtering, we use the relation between $p$ and $f$ through $M$ to define the likelihood of a mapping. We observe that for maximum likelihood solutions given $M$, fixing a mapping $f$ automatically fixes an optimal $p=\text{Params}(f,M)$. Thus, as before, we focus our attention on the mappings, implicitly finding the maimum likelihood $p$ as well. The following lemma and its proof sketch capture this idea.
\begin{lemma}[Likelihood of a mapping]
We have
\begin{align*}
\underset{f,p}\max \Pr(M|f,p)=\underset{f}\max \Pr(M|f,\text{Params}(f,M)) \\
\text{where \ \ } \exists\text{Params}(f,M)=\underset{p}{\arg\!\max} \Pr(p|f,M)
\end{align*}
\end{lemma}
\begin{proof}
%\smallskip
%\noindent {\bf $\text{Params}\mathbf{(f,M)}$.}
Given mapping $f$ and evidence $M$, we can calculate the worker response probability matrix $p=\text{Params}(f,M)$ as follows. Let the $i^{th}$ dimension of the response set of any item $I$ by $M_i(I)$. That is, if $M(I)=(v_R,\ldots,v_1)$, then $M_i(I)=v_i$. Let $\mathbf{I_i}\subseteq\mathbf{I}\ni f(I)=i\forall i,I\in\mathbf{I_i}$. Then, $p(i,j)=\frac{\sum_{I\in \mathbf{I_j}}M_{i}(I)}{m|I_j|}\forall i,j$. Intuitively, $\text{Params}(f,M)(i,j)$ is just the fraction of times a worker responded $i$ to an item that is mapped by $f$ to a value of $j$. Similar to Lemma~\ref{params}, we can show that $\text{Params}(f,M) = \underset{p}{\arg\!\max} $ \\ $\Pr(p|f,M)$.
Consequently, it follows that $$\underset{f,p}\max \Pr(M|f,p)=\underset{f}\max \Pr(M|f,\text{Params}(f,M)) \qed$$
\end{proof}

Denoting the likelihood of a mapping, $\Pr(M|f,\text{Params}(f,M))$, as $\Pr(M|f)$, our maximum likelihood rating problem is now equivalent to that of finding the most likely mapping. Thus, we wish to solve for $\underset{f}{\arg\!\max} \Pr(M|f)$.
%\begin{lemma}[Likelihood of a Mapping]
%\label{lemma:likelihood-rating}
%Let $f\in\mathbf{F}$ be any mapping and $M$ be the given response set on $\mathbf{I}$. We have, \[\arg\!\max_{f,p}\Pr(M|f,p)=\arg\!\max_{f}\Pr(M|f,\text{Params}(f,M))\]. %Also, we can evaluate $\Pr(M|f,\text{Params}(f,M))$ in $O(m|\mathbf{I}|)$.
%\end{lemma}
%\smallskip
%\noindent {\bf Maximum Likelihood Mapping.}

%\subsection{Intuition}
%\label{sec:rating-intuition}
\subsection{Algorithm}
\label{sec:rating-algo}
Now, we generalize our idea of bucketized, dominance-consistent mappings from Section \ref{algo} to find a maximum likelihood solution for the rating problem. \techreport{Although we primarily present the intuition below, we formalize our dominance relation and consistent-mappings in Section \ref{sec:rating-dom_prop} and further prove some interesting properties.}

\smallskip
\noindent {\bf Bucketizing.}
For every item, we are given $m$ worker responses, each in $1,2,\ldots,R$. It can be shown that there are $R+m-1 \choose R-1$ different possible worker response sets, or buckets. The bucketizing idea is the same as before: items with the same response sets can be treated identically and should be mapped to the same values. So we only consider mappings that give the same rating score to all items in a common response set bucket.

\smallskip
\noindent {\bf Dominance Ordering.}
Next we generalize our dominance constraint. Recall that for filtering with $m$ responses per item, we had a total ordering on the dominance relation over response set buckets, $(m,0)>(m-1,1)>\ldots>(1,m-1)>(0,m)$ where no dominated bucket could have a higher score (``1'') than a dominating bucket (``0''). Let us consider the simple example where $R=3$ and we have $m=3$ worker responses per item. %That is, each item has a response set in $\{1,2,3\}^3$.
Let $(i,j,k)$ denote the response set where $i$ workers give a score of $``3''$, $j$ workers give a score of $``2''$ and $k$ workers give a score of $``1''$. Since we have $3$ responses per item, $i+j+k=3$. Intuitively, the response set $(3,0,0)$ dominates the response set $(2,1,0)$ because in the first, three workers gave items a score of ``3'', while in the second, only two workers give a score of ``3'' while one gives a score of ``2''. Assuming a ``reasonable'' worker behavior, we would expect the value assigned to the dominating bucket to be at least as high as the value assigned to the dominated bucket.
%We show formally in our technical report that this dominance condition is sufficient to guarantee optimality (---cite tech report here---).
Now consider the buckets $(2,0,1)$ and $(1,2,0)$. For items in the first bucket, two workers have given a score of ``3'', while one worker has given a score of ``1''. For items in the second bucket, one worker has given a score of ``3'', while two workers have given a score of ``2''. Based solely on these scores, we cannot claim that either of these buckets dominates the other. So, for the rating problem we only have a \emph{partial} dominance ordering, which we can represent as a DAG. We show the \emph{dominance-DAG} for the $R=3,m=3$ case in Figure \ref{fig:dag-3,3}.
%{\centering
%\begin{figure}
%\includegraphics[scale=0.3]{3,3-Dominance-Dag.pdf}
%\caption{Dominance-DAG for 3 workers and scores in $\{1,2,3\}$}
%\label{fig:dag-3,3}
%\end{figure}
%}
\begin{figure}
\centering
\includegraphics[scale=0.35]{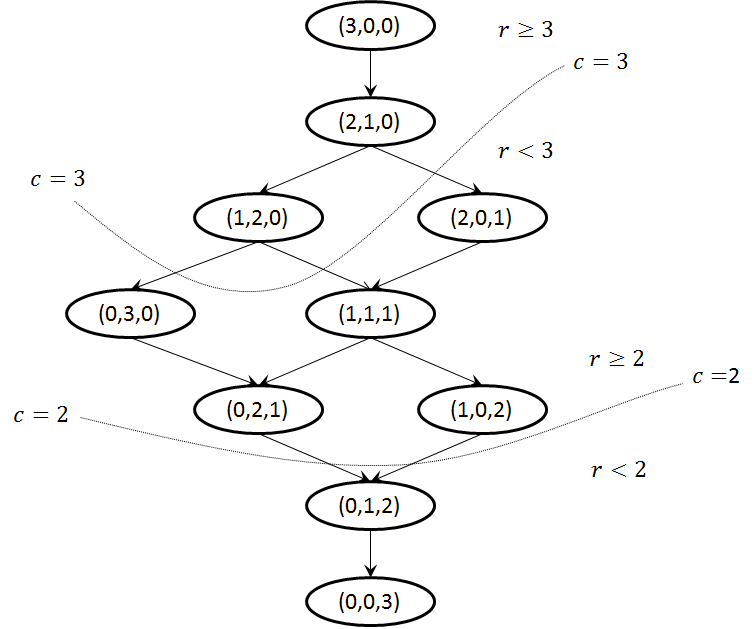}
\vspace{-5pt}
\caption{Dominance-DAG for 3 workers and scores in $\{1,2,3\}$}
\vspace{-15pt}
\label{fig:dag-3,3}
\end{figure}
%
%In the DAG, a node $(i,j,k)$ represents the bucket corresponding to items where $i$ workers give a score of $3$, $j$ workers give a score of $2$ and $k$ workers give a score of $1$ ($i+j+k=3$). We naturally expect items in bucket $(3,0,0)$ to receive the highest rating across buckets. Similarly, we expect bucket $(1,2,0)$ to receive at least as high a rating as bucket $(1,1,1)$, although it cannot directly be compared against bucket $(2,0,1)$.
For arbitrary $m,R$, we can define the following dominance relation.
\begin{definition}[Rating Dominance]
\label{def:dominance}
Bucket $B_1$ with response set $(v_R^1,v_{R-1}^1,\ldots,v_1^1)$ \emph{dominates} bucket $B_2$ with response set $(v_R^2,v_{R-1}^2,\ldots,v_1^2)$ if and only if $\exists 1<r'\leq R\ni (v_r^1=v_r^2\forall r\not\in\{r',r'-1\})$ and $(v_{r'}^1=v_{r'}^2+1)\land(v_{r'-1}^1=v_{r'-1}^2-1)$.
\end{definition}
Intuitively, a bucket $B_1$ dominates $B_2$ if increasing the score given by a single worker to $B_2$ by $1$ makes its response set equal to that of items in $B_1$. Note that a bucket can dominate multiple buckets, and a dominated bucket can have multiple dominating buckets, depending on which worker's response is increased by $1$. For instance, %from the $R=3,m=3$ example, %bucket $(3,0,0)$ dominates bucket $(2,1,0)$ because the second bucket can be transformed to the first by increasing the score of ``2'' given by one worker to ``3''.
%Similarly,
in Figure \ref{fig:dag-3,3}, bucket $(2,1,0)$ dominates both $(2,0,1)$ (increase one score from ``1'' to ``2'') and $(1,2,0)$ (increase one score from ``2'' to ``3''), both of which dominate $(1,1,1)$. %This dominance relation induces a partial order over the buckets for arbitrary $m,R$, which we represent as our dominance-DAG.

\smallskip
\noindent {\bf Dominance-Consistent Mappings.}
As with filtering, we consider the set of mappings satisfying both the bucketizing and dominance constraints, and call them \emph{dominance-consistent mappings}.

Dominance-consistent mappings can be represented using {\em cuts} in the dominance-DAG. To construct a dominance-consistent mapping, we split the DAG into at most $R$ partitions such that no parent node belongs to an intuitively ``lower'' partition than its children. Then we assign ratings to items in a top-down fashion such that all nodes within a partition get a common rating value lower than the value assigned to the partition just above it. Figure \ref{fig:dag-3,3} shows one such dominance-consistent mapping corresponding to a set of cuts. A cut with label $c=i$ essentially partitions the DAG into two sets: the set of nodes above all receive ratings $\geq i$ while all nodes below receive ratings $<i$. %We also show one example of an invalid cut, $c'$, because it places the child node $(0,1,2)$ in a higher partition than its parent $(1,0,2)$, thereby violating our dominance condition.
To find the most likely mapping, we sort the items into buckets and look for mappings over buckets that are consistent with the dominance-DAG. We use an iterative top-down approach to enumerate all consistent mappings. First, we label our nodes in the DAG from $1\ldots {R+m-1\choose R-1}$ according to their topological ordering, with the root node starting at $1$. In the $i^{th}$ iteration, we assume we have the set of all possible consistent mappings assigning values to nodes $1\ldots i-1$ and extend them to all consistent mappings over nodes $1\ldots i$. When the last ${R+m-1\choose R-1}^{th}$ node has been added, we are left with the complete set of all dominance-consistent mappings.
\papertext{Full details of our algorithm, appear in the online technical report \cite{optTR}.}% (--cite here--).

\techreport{
\begin{algorithm}[h!]
\caption{Dominance-Consistent Mappings}
\label{rating_algo}
\begin{algorithmic}[1]
\STATE $I:=\text{Input Item-set}$
\STATE $M:=\text{Input Evidence Matrix}$
\STATE $F:=\{\}$ \COMMENT{Different dominance-consistent (mappings)}
%\STATE $s:=\{\}$ \COMMENT{Selectivities corresponding to cut-functions}
\STATE $p:=\{\}$ \COMMENT{Worker matrices corresponding to mappings}
\STATE $Likelihood:=\{\}$ \COMMENT{Likelihoods corresponding to mappings}
\STATE Construct $V,E=$ Dominance-DAG
\\ \COMMENT{Enumerating consistent mappings}
\FOR{$v$ in BFS($V$)}
    \STATE{(expand dominance-DAG by BFS)}
    \FOR{$f$ in $F$}
        \STATE{(expand old mappings to include $v$)}
        \STATE{$lower := \min_{v'\in\text{parents}(v)}f(v')$}
        \STATE{$upper := \max_{v'\in\text{parents}(v)}f(v')$}
        \FOR{$i$ in $lower$ to $upper$}
            \STATE $f_{new}[i]:=f\cup \{v=i\}$
            \STATE{$F$.add($f_{new}[i]$)}\COMMENT{add new mappings corresponding to the dominance-consistent possible values for $v$}
        \ENDFOR
        \STATE{Delete $f$}\COMMENT{delete old mappings that only mapped nodes $1,2,\ldots,v-1$}
    \ENDFOR
\ENDFOR
\FOR{$f$ in $F$}
    \STATE{$p[f]:=\text{Params}(f,M)$}
    \STATE{$Likelihood[f]:=\Pr(M|f,p[f])$}
\ENDFOR
\STATE $f^*:=\arg\!\max Likelihood[f]$
\STATE RETURN{$(f^*,p[f^*]])$}
\end{algorithmic}
\end{algorithm}
}

%\smallskip
%\noindent {\bf Maximum likelihood guarantee.}
As with the filtering problem, we \techreport{can} show \papertext{in \cite{optTR}} %(---cite full tech report here--)
that an exhaustive search of the dominance-consistent mappings under this dominance DAG constraint gives us a global maximum likelihood mapping across a much larger space of reasonable mappings. Suppose we have $n$ items, $R$ rating values, and $m$ worker responses per item. The number of buckets of possible worker response sets (nodes in the DAG) is $R+m-1 \choose R-1$. Then, the number of unconstrained mappings is $R^n$ and number of mappings with just the bucketizing condition, that is where items with the same response sets get assigned the same value, is $R^{R+m-1 \choose R-1}$. We enumerate %the number of our dominance-consistent mappings after additionally applying the dominance constraint on the bucketized mappings and compare these numbers
a sample set of values in Table \ref{table:num_mappings} for $n=100$ items. We see that the number of dominance-consistent mappings is significantly smaller than the number of unconstrained mappings. The fact that this greatly reduced set of intuitive mappings contains a global maximum likelihood solution displays the power of our approach. %We assume $n=100$ here, but typically the number of items is much higher which
Furthermore, the number of items may be much larger, which would make the number of unconstrained mappings exponentially larger.
\begin{table}
\centering
\small
\begin{tabular}{ | c | c | c | c | c | }
\hline
$R$ & $m$ & Unconstrained & Bucketized & Dom-Consistent \\ \hline
3 & 3 & $10^{47}$ & $6\times 10^4$ & 126 \\ \hline
3 & 4 & $10^{47}$ & $10^7$ & 462 \\ \hline
3 & 5 & $10^{47}$ & $10^{10}$ & 1716 \\ \hline
4 & 3 & $10^{60}$ & $10^{12}$ & $2.8\times 10^4$ \\ \hline
4 & 4 & $10^{60}$ & $10^{21}$ & $2.7\times 10^{6}$ \\ \hline
5 & 2 & $10^{69}$ & $10^{10}$ & $2.8\times 10^4$ \\ \hline
5 & 3 & $10^{69}$ & $10^{24}$ & $1.1\times 10^8$ \\ \hline
\end{tabular}
\vspace{-10pt}
%\end{center}
\caption{Number of Mappings for $n=100$ items\label{table:num_mappings}
\vspace{-15pt}
}
\end{table}

%\subsection{Algorithm}
%In Section \ref{sec:rating-intuition}, we described the dominance partial order for the $m=3$, $R=3$ problem.

\subsection{Experiments}
\label{sec:rating-exp}
We perform experiments using simulated workers and synthetic data for the rating problem using a setup similar to that described in Section \ref{sec:filter-exp-simulated}. Since our results and conclusions are similar to those in the filtering section, we show the results from one representative experiment and refer interested readers to \papertext{our full technical report for a more extensive experimental analysis.} \techreport{the appendix, Section \ref{sec:appendix-rating-experiments} for further results.}

%In Figures \ref{fig:rating-like} and \ref{fig:rating-dist_wtd},
\smallskip
\noindent {\bf Setup.}
We use 1000 items equally distributed across true ratings of $\{1,2,3\}$ ($R=3$). We randomly generate worker response probability matrices and simulate worker responses for each item to generate one response set $M$. We plot and compare various quality metrics of interest, for instance, the likelihood of mappings and quality of predicted item ratings along the y-axis, and vary the number of worker responses per item, $m$, along the x-axis. Each data point in our plots corresponds to the outputs of corresponding algorithms averaged across 100 randomly generated response sets, that is, 100 different $M$s.
The initializations of the EM algorithms correspond to \xspace{the worker response probability matrices
$ EM(1)=
 \begin{bmatrix}
 0.6 & 0.33 & 0.07\\
 0.33 & 0.34 & 0.33\\
 0.07 & 0.33 & 0.6
 \end{bmatrix}$
 ,
$EM(2)=
 \begin{bmatrix}
 0.34 & 0.33 & 0.33\\
 0.33 & 0.34 & 0.33\\
 0.33 & 0.33 & 0.34
 \end{bmatrix}
$, and
$EM(3)=
 \begin{bmatrix}
 0.07 & 0.33 & 0.6\\
 0.33 & 0.34 & 0.33\\
 0.6 & 0.34 & 0.07
 \end{bmatrix}
 $}.
 Intuitively, $EM(1)$ starts by assuming that workers have low error rates, $EM(2)$ assumes that workers answer questions uniformly randomly, and $EM(3)$ assumes that workers have high (adversarial) error rates. As in Section \ref{sec:filter-exp-simulated}, $EM(*)$ picks the most likely from the three different EM instances for each response set $M$.

\smallskip
\noindent {\bf Likelihood.}
Figure \ref{fig:rating-like} plots the likelihoods (on a natural log scale) of the mappings output by different algorithms along the y-axis. We observe that the likelihoods of the mappings returned by our algorithm, $OPT$, are significantly higher than those of any of the EM algorithms. %For example, consider $m=5$: the average (normalized) log likelihood of mappings returned by our algorithm is $-701$, while the average log likelihood of mappings returned by $EM(*)$ (in this case the best EM instance), is $-722$. Thus we see, that for $m=5$ our algorithm finds mappings that are on average $9$ orders of magnitude more likely. As with filtering, the gap between the performance of our algorithm and the EM instances decreases as the number of workers increases.
For example, consider $m=5$: we observe that our algorithm finds mappings that are on average 9 orders of magnitude more likely than those returned by $EM(*)$ (in this case the best EM instance). As with filtering, the gap between the performance of our algorithm and the EM instances decreases as the number of workers increases.
%{\centering
%\begin{figure}
%\includegraphics[scale=0.4]{rating-likelihood.png}
%\caption{Likelihood, $R=3$}
%\label{fig:rating-like}
%\end{figure}
%}

\begin{figure}
\centering
\includegraphics[scale=0.3]{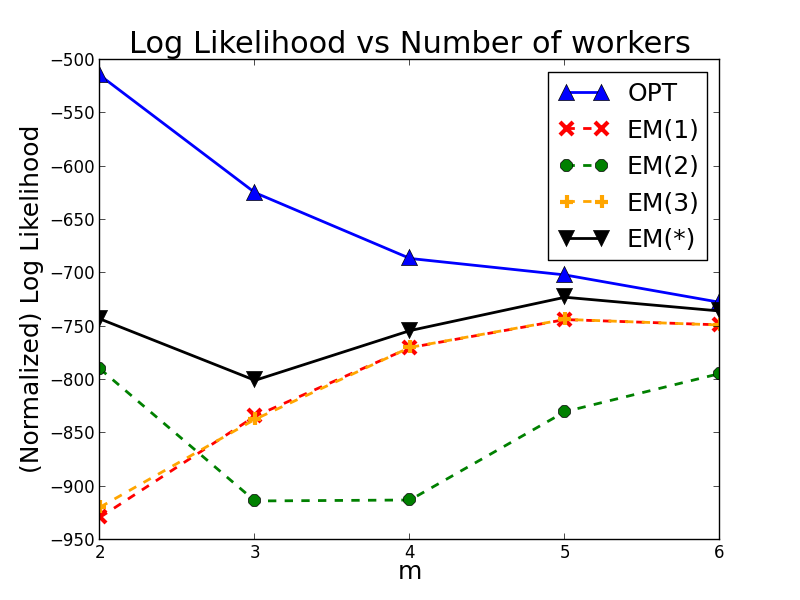}
\vspace{-15pt}
\caption{Likelihood, $R=3$}
\label{fig:rating-like}
\vspace{-15pt}
\end{figure}

\smallskip
\noindent {\bf Quality of item rating predictions.}
In Figure \ref{fig:rating-dist_wtd} we compare the predicted ratings of items against the true ratings used to generate each data point. We measure a weighted score based on how far the predicted value is from the true value; a correct prediction incurs no penalty, a predicted rating that is $\pm 1$ of the true rating of the item incurs a penalty of $1$ and a predicted rating that is $\pm 2$ of the true rating of the item incurs a penalty of $2$. We normalize the final score by the number of items in the dataset. For our example, each item can result in a maximum penalty of 2, therefore, the computed score is in $[0,2]$, with a lower score implying more accurate predictions. Again, we observe that in spite of optimizing for, and providing a global maximum likelihood guarantee, our algorithm predicts item ratings with a high degree of accuracy. %Interestingly, in this case, we observe that $EM(2)$ does better than $EM(1)$. We believe this is also an artifact of higher dimensionality of the rating problem; for the binary case, with fewer non-optimal local maxima, $EM(1)$ often converges naturally to the (simple) global optima. $EM(2)$, on the other hand sometimes converges close to the global optima, similar to $EM(1)$, while in other cases it converges in the opposite direction, similar to $EM(3)$. For $R=3$, this is not the case. While we observe that $EM(1)$ does indeed predict values better than $EM(3)$ for better than random workers, \ads{better argument?}

%{\centering
%\begin{figure}
%\includegraphics[scale=0.4]{rating-dist_wtd.png}
%\caption{Item prediction (distance weighted), $R=3$}
%\label{fig:rating-dist_wtd}
%\end{figure}
%}

\begin{figure}
\centering
\vspace{-5pt}
\includegraphics[scale=0.3]{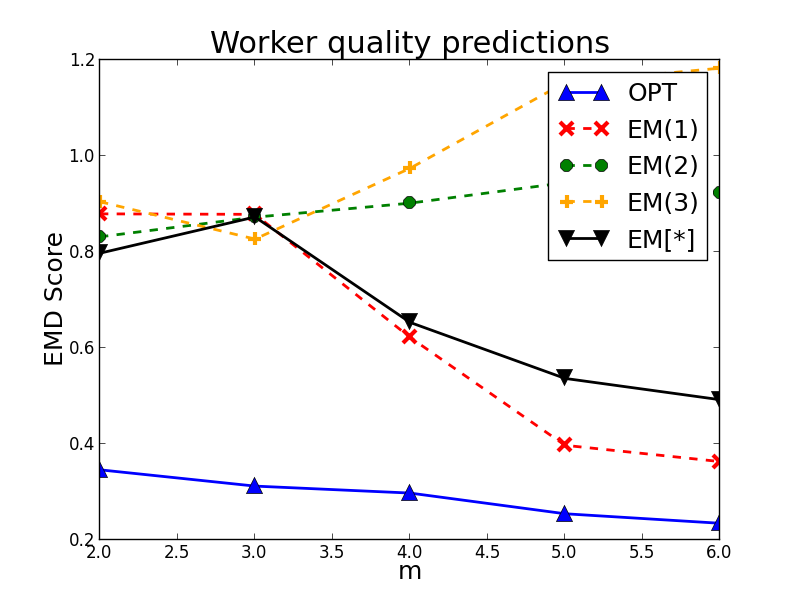}
\vspace{-10pt}
\caption{Item prediction (distance weighted), $R=3$}
\label{fig:rating-dist_wtd}
\vspace{-15pt}
\end{figure}

%\ads{
Comparing these results to those in Section \ref{sec:filter-exp-simulated}, our gains for rating are significantly higher because the number of parameters being estimated is much higher, and the EM algorithm has more ``ways'' it can go wrong if the parameters are initialized incorrectly. That is, with a higher dimensionality we expect that EM converges more often to non-optimal local maxima.
%}

\papertext{We perform further experiments on synthetic data across a wider range of input parameters and report our complete findings in our full technical report \cite{optTR}. }% (---cite here---).
We observe that in spite of being optimized for likelihood, our algorithm performs well, often beating EM, for different metrics of comparison on the predicted item ratings and worker response probability matrices.

\subsection{Formalizing dominance}
\label{sec:rating-dom_prop}
In this section, we formalize our dominance relation and prove that it is in fact a partial order, and more specifically, a lattice. Let $\mathbf{V}$ be the set of all possible item response sets. Recall from Section \ref{sec:rating-algo} that $|\mathbf{V}|={R+m-1 \choose R-1}$. Definition \ref{def:dominance} defines the notion of one response set just dominating, or {\em covering} another response set. If response set $V_1$ covers response set $V_2$ under Definition \ref{def:dominance}, we write $V_1\succ V_2$. We extend that definition to include transitive dominance below.
\begin{definition}[Transitive dominance]
\label{def:trans}
Let $I_1,I_2,I_3$ be any three items with response sets $M(I_1)=V_1,M(I_2)=V_2,M(I_3)=V_3$. We define the transitive dominance relation ($\succ_{t}$) on sets $\mathbf{I}$ and $\mathbf{V}$ as follows:
\begin{enumerate}
\item $I\succeq_t I\forall I\in\mathbf{I}$ and $V\succeq_t V\forall V\in\mathbf{V}$
\item If $I_1\succ I_2$, then $I_1\succeq_{t}I_2$. Similarly, $V_1\succ V_2\Rightarrow V_1\succeq_{t}V_2$
\item If $I_1\succeq_{t}I_2\land I_2\succeq_{t}I_3$, then $I_1\succeq_{t}I_3$. Similarly, $V_1\succeq_{t}V_2\land V_2\succeq_{t}V_3\Rightarrow V_1\succeq_{t}V_3$
\end{enumerate}
If $I_1\succeq_t I_2\land I_1\neq I_2$, we write $I_1\succ_t I_2$. Intuitively, the transitive dominance relation constitutes the transitive closure of the dominance relation.
\end{definition}

\begin{definition}[Partial Ordering]
\label{defn:poset}
Let $\geq$ be a binary relation on set $S$. We say that $\geq$ defines a partial order on $S$ if the following are satisfied for all $x,y,z\in S$:
\begin{enumerate}
\item Reflexivity: $x\geq x$.
\item Antisymmetry: If $x\geq y$ and $y\geq x$, then $x=y$.
\item Transitivity: If $x\geq y$ and $y\geq z$, then $x\geq z$.
\end{enumerate}
\end{definition}

We show below that our dominance relation imposes a partial ordering on the set of possible item response sets. To do so, we first introduce the idea of a {\em cumulative distribution}, and use it to characterize our transitive dominance relation.

\begin{lemma}[Cumulative Distribution]
\label{lemma:cum_dist}
Let $A=$ \\$(a_R,a_{R-1},\ldots,a_1),B=(b_R,b_{R-1},\ldots,b_1)\in\mathbf{V}$ be any two realizations such that $A\succeq_t B$. Let $Cum(A)=(A_R,A_{R-1},\ldots,A_1)$ and $Cum(B)=(B_R,B_{R-1},\ldots,B_1)$ be their cumulative distribution functions, where $A_j=\sum_{i=1}^ja_i$ and $B_j=\sum_{i=1}^jb_i$. Then, $A_i\geq B_i\forall i\in 1\text{ to }n$.
\end{lemma}
\noindent{\bf Proof.} Let $A=X_1\succ X_2\succ \ldots \succ X_k=B$ be a sequence of realizations just dominating (or covering) the next. Intuitively, to move from $X_j=(x_{j,1},x_{j,2},\ldots,x_{j,n})$ to $X_{j+1}=$ \\$(x_{j+1,1},x_{j+1,2},\ldots,x_{j+1,n})$, we need to shift one vote from some bucket $k$ to $k+1$. That is, $x_{j,k}\rightarrow x_{j,k}-1$ and $x_{j,k+1}\rightarrow x_{j,k+1}+1$ (follows from Definition \ref{def:dominance}).
The path from $A=X_1$ to $B=X_k$ can be represented by a sequence of such unit vote moves towards higher ratings. Let the total number of votes shifted from bucket $i$ to bucket $i+1$ in the entire path $<X_1,X_k>$ be $\delta_i$.
Then, $B=(b_R,b_{R-1},\ldots,b_1)=(a_R-\delta_R+\delta_{R-1},\dots,a_3-\delta_3+\delta_2,a_2-\delta_2+\delta_1,a_1-\delta_1)$. Note that we are constrained by $\delta_R=0$ (cannot move any further than highest bucket, $R$) and $0\leq \delta_i\leq a_i+\delta_{i-1}\forall i<R$. Now, it is easy to verify that $Cum(B)=(B_R,B_{R-1},\ldots,B_1)=(A_R-\delta_R,A_{R-1}-\delta_{R-1},\ldots,A_1-\delta_1)$. Since $\delta_i\geq 0\forall i$, we have $A_i=B_i+\delta_i\Rightarrow A_i\geq B_i\forall i$.
\qedr\\

Lemma \ref{lemma:cum_dist} gives us a way to represent descendants in our dominance ordering using the cumulative distribution function. We use this idea to prove that our dominance relation is a partial order on the set of realizations, and more specifically, a lattice.

\begin{lemma}[Partial Order]
\label{lemma:poset}
The relation $\succeq_t$ on the set of items $\mathbf{I}$ or the set of response sets, $\mathbf{V}$ defines a partial ordering on the respective domains.
\end{lemma}
\noindent{\bf Proof.} We show that $(\mathbf{V},\succeq_t)$ is a partial order. %A similar proof holds for $(\mathbf{I},\succeq_t)$.\\
From Definition \ref{def:trans}, we have $V\succeq V\forall V\in\mathbf{V}$. So, our dominance relation $\succeq$ is reflexive.\\
Let $A\succeq_t B$ and $B\succeq_t A$ for some $A=(a_R,a_{R-1},\ldots,a_1),B=(b_R,b_{R-1},\ldots,b_1)\in\mathbf{V}$. Consider the cumulative distribution function, $Cum(A)=(A_R,A_{R-1},\ldots,A_1)$ and $Cum(B)=$ \\$(B_R,B_{R-1},\ldots,B_1)$, where $A_j=\sum_{i=1}^ja_i$ and $B_j=\sum_{i=1}^jb_i$. From Lemma \ref{lemma:cum_dist}, we have $A\succeq_t B\Rightarrow A_i\geq B_i\forall i$. Similarly, $B\succeq_t A\Rightarrow B_i\geq A_i\forall i$. Combining, we have $A_i=B_i\forall i\Rightarrow a_i=b_i\forall i$. Therefore, $A=B$ and $\succeq_t$ is antisymmetric.\\
From Definition \ref{def:trans}, we have $V_1\succeq_{t}V_2\land V_2\succeq_{t}V_3\Rightarrow V_1\succeq_{t}V_3\forall V_1,V_2,V_3\in\mathbf{V}$. So, $\succeq_t$ is transitive.\\
Therefore, the relation $\succeq_t$ is a partial order.
\qedr\\

We further show that the partial order imposed by our dominance relation, $\succeq_t$, is in fact a lattice. A lattice can be defined as follows.

\begin{definition}[Lattice]
A partially ordered set $(\mathbf{V},\succeq)$ is a lattice if it satisfies the following properties:
\begin{enumerate}
\item $\mathbf{V}$ is finite.
\item There exists a maximum element $V^*\in\mathbf{V}$ such that $V^*\succeq V\forall V\in\mathbf{V}$.
\item Every pair of elements has a greatest lower bound (meet), that is, $\forall V_1,V_2\in\mathbf{V}\exists V'\in \mathbf{V}\ni (V_1,V_2\succeq V')\land(\not\exists V\in\mathbf{V})\ni V_1,V_2\succeq V\succ V'$.
\end{enumerate}
\end{definition}

We now show that our partial ordering $\succeq_t$ on the set of realizations $\mathbf{V}$ is a lattice.

\begin{theorem}[Lattice Proof]
The dominance partial ordering $(\mathbf{V},\succeq_t)$ is a lattice.
\end{theorem}
\noindent{\bf Proof.}
It is easy to see that our set of realizations is finite ($|\mathbf{V}|={R+M-1 \choose R-1}$.
Next, consider the element $V^*=(M,0,\ldots,0)$. We have, $V^*\succeq_t V\forall V\in\mathbf{V}\setminus V^*$.
Therefore, all that remains to be shown is that every pair of realizations has a unique greatest lower bound.

Let $A=(a_R,a_{R-1},\ldots,a_1),B=(b_R,b_{R-1},\ldots,b_1)\in\mathbf{V}$ be any two realizations with cumulative distribution functions $Cum(A)=(A_R,A_{R-1},\ldots,A_1)$ and $Cum(B)=(B_R,B_{R-1},\ldots,B_1)$, $A_j=\sum_{i=1}^ja_i$ and $B_j=\sum_{i=1}^jb_i$. Let $D=(d_R,d_{R-1},\ldots,d_1)$ be any common descendant (or lower bound) of $A,B$, that is $A,B\succeq_t D$. Let $Cum(D)=(D_R,D_{R-1},\ldots,D_1)$. From Lemma \ref{lemma:cum_dist} it follows that we can find $\delta=(\delta_R,\ldots,\delta_1)$ and $\delta'=(\delta'_R,\ldots,\delta'_1)$ such that $Cum(D)=(A_R-\delta_R,\ldots,A_1-\delta_1)=(B_R-\delta'_R,\ldots,B_1-\delta'_1)$ where $0\leq \delta_i\leq A_i$, $0\leq \delta'_i\leq B_i$ $\forall i<R$, and $\delta_R=\delta'_R=0$.

Choose $\delta^*_i=\max(0,A_i-B_i)$ and $\delta'^*_i=\max(0,B_i-A_i)$. That is, if $A_i<B_i$, we have $\delta^*_i=0,\delta'^*_i=B_i-A_i$ and if $B_i\leq A_i$, $\delta^*_i=A_i-B_i,\delta'^*_i=0$. Let $D^*$ be the lower bound to $A,B$ constructed from $\delta^*,\delta'^*$ such that  $Cum(D^*)=Cum(A)-\delta^*=Cum(B)-\delta'^*$ and $0\leq \delta^*_i\leq A_i$, $0\leq \delta'^*_i\leq B_i$ $\forall i<R$, and $\delta^*_R=\delta'^*_R=0$. So, $D^*$ is a common lower bound of $A,B$. We claim that $D^*$ is in fact the greatest lower bound of $A,B$. We prove our claim in two steps.

First, let $D$ be any strict upper bound or ancestor of $D^*$. We show that $D$ cannot be a common lower bound to $A,B$. From Lemma \ref{lemma:cum_dist}, we have $Cum(D)=Cum(D^*)+\Delta$ where $\Delta=(\Delta_R,\ldots,\Delta_1)$, such that $\Delta_i\geq 0\forall i$ and $\Delta_{k}>0$ for some $k$. Now $D_k=D^*_k+\Delta_k$. From our construction of $D^*$, we have $D^*_k=A_k-\delta^*_k=B_k-\delta'^*_k$ where one of $\{\delta^*_k,\delta'^*_k\}$ is $0$. Without loss of generality, suppose $\delta_k=0$. Then, $D_k=A_k+\Delta_k$. Since $D_k>A_k$, by Lemma \ref{lemma:cum_dist} $A\not\succeq_t D$.

Second, let $D\neq D^*$ be any lower bound to $A,B$. We show that $\exists D'\succ_t D$ such that $D'$ is also a lower bound to $A,B$. Let $D$ be the lower bound constructed from $\delta,\delta'$, that is, $D_i=A_i-\delta_i=B_i-\delta'_i\forall i$. Now, $\exists k\ni \delta_k,\delta'_k\neq 0$ (otherwise $D=D^*$). Construct $D'$ such that $D'_i=D_i\forall i\neq k$ and $D'_k=D_k+\min(\delta_k,\delta'_k)$. It is easy to verify that $D'\succ_t D$ and $D'$ is a lower bound of $A,B$.

Combining the facts that (a) $D^*$ is lower bound of $A,B$ with no ancestor that is also a lower bound of $A,B$, and (b) Any other lower bound of $A,B$ can be shown to have an ancestor that is also a lower bound of $A,B$, we have $D^*$ is greatest lower bound of $A,B$.\\
This completes our proof.
\qedr\\

Next, we formally define dominance-consistent mappings that are consistent with the above intuition.

\begin{definition}[Dominance-Consistent Mapping]
\label{consistent_mapping}
We call a function $f^\delta\in\mathbf{F}:\mathbf{I}\rightarrow [1,R]$ a Dominance-Consistent mapping if it satisfies the following properties:
\begin{enumerate}
\item Let $I_1,I_2\in\mathbf{I}$ be any two items. If $M(I_1)=M(I_2)$, then $f^\delta(I_1)=f^\delta(I_2)$.
\item Let $V_1,V_2\in\mathbf{V}\ni M(I_1)=V_1\succ_t V_2=M(I_2)$. Then, $f^\delta(I_1)\leq f^\delta(I_2)$.
\end{enumerate}
We denote the set of all dominance consistent mappings by $\mathbf{F^\delta}\subseteq\mathbf{F}$.
\end{definition}
Intuitively, the first property ensures that a consistent mapping assigns the same bucket to items with the same observed response sets. The second property states that the mapping is consistent with the transitive dominance relation, that is, an item with a better response  set is mapped to at least as high a bucket as an item with a worse response set. Note that it is crucial to use the transitive dominance relation, and not just the dominance relation when defining consistent mappings to preserve our intuition. Otherwise, consider an example where there exist two items $I_1,I_2\in\mathbf{I}$ such that $M(I_1)\succ_t M(I_2)$ and yet, $\not\exists I\in\mathbf{I}$ such that $M(I_1)\succ M(I)$. It would then be possible to construct a consistent mapping, $f$, with $f(I_1)>f(I_2)$, which would violate the intuition behind consistent mappings.

%!TEX root=Parameter Estimation.tex

\section{Extensions}
\label{sec:extensions}

In this section we discuss the generalization of our bucketizing and dominance-based approach to some extensions of the filtering and rating problems. Recall our two major assumptions: (1) every item receives the same number ($m$) of responses, and (2) all workers are randomly assigned and their responses are drawn from a common distribution, $p(i,j)$. We now relax each of these requirements and describe how our framework can be applied.

\subsection{Variable number of responses}
\label{sec:extensions-variable}
Suppose different items may receive different numbers of worker responses, \eg because items are randomly chosen, or workers %. This could easily happen in cases where items are randomly chosen, with some items being shown several items, while some others are rarely shown. A skew in the number of responses to items could also arise from the fact that often the workers may choose
choose some questions preferentially over others.
% and they themselves decide which, and how many questions to answer.
Note in this section we are still assuming that all workers have the same response probability matrix $p$.

For this discussion we restrict ourselves to the filtering problem; a similar analysis can be applied to rating. Suppose each item can receive a maximum of $m$ worker responses, with different items receiving different numbers of responses. Again, we bucketize items by their response sets and try to impose a dominance-ordering on the buckets. %As before, we use the response set $(k-j,j)$ to denote the bucket of items that receive $k-j$ ``1'' responses and $j$ ``0'' responses.
Now, instead of only considering response sets of the form $(m-j,j)$, we consider arbitrary $(i,j)$. Recall that a response set $(i,j)$ denotes that an item received $i$ ``1'' responses and $j$ ``0'' responses. We show the imposed dominance ordering in Figure \ref{fig:ext-1}.

We expect an item that receives $i$ ``1'' responses and $j$ ``0'' responses to be more likely to have true value ``1'' than an item with $i-1$ ``1'' responses and $j$ ``0'' responses, or an item with $i$ ``1'' responses and $j+1$ ``0'' responses. So, we have the dominance relations $(i,j)>(i-1,j)$ where $i\geq 1, j\geq 0, i+j\leq m$, and $(i,j)>(i,j+1)$ with $i,j\geq 0, i+j+1\leq m$.
Note that the dominance ordering imposed in Section \ref{sec:filtering}, $(m,0)>(m-1,1)>\ldots>(0,m)$, is implied transitively here. For instance, $(m,0)>(m-1,0)\land (m-1,0)>(m-1,1)\Rightarrow (m,0)>(m-1,1)$. Also note that this is a partial ordering as certain pairs of buckets, $(0,0)$ and $(1,1)$ for example, cannot intuitively be compared.

Again, we can reduce our search for the maximum likelihood mapping to the space of all bucketized mappings consistent with this dominance (partial) ordering. That is, given item set $\mathbf{I}$ and response set $M$, we consider mappings $f:\mathbf{I}\rightarrow \{0,1\}$, where $M(I_1)=M(I_2)\Rightarrow f(I_1)=f(I_2)$ and $M(I_1)>M(I_2)\Rightarrow f(I_1)\geq f(I_2)$. We show two such dominance consistent mappings, $f_i$ and $f_j$ in Figure \ref{fig:ext-1}. Mapping $f_i$ assigns all items with at least $i$ ``1'' worker responses to a value of 1 and the rest to a value of 0. Similarly, mapping $f_j$ assigns all items with at most $j$ ``0'' responses a value of 1 and the rest a value of 0. We can construct a third dominance-consistent mapping $f_{ij}$ from a conjunction of these two: $f_{ij}(I)=1$ if and only if $f_i=1\land f_j=1$, that is, $f_{ij}$ assigns only gives those items that have at least $i$ ``1'' worker responses and at most $j$ ``0'' responses, a value of 1. We can now describe $f_i$ and $f_j$ as special instances of the dominance-consistent mapping $f_{ij}$ when $j=m$ and $i=0$ respectively.

We claim that all dominance-consistent mappings for this setting can be described as the union of different $f_{ij}$s for a set of $0\leq i,j\leq m$, for a total of $O(2^m)$ dominance-consistent mappings. Note that although this expression is exponential in the maximum number of worker responses per item, $m$, for most practical applications this is a very small constant. \papertext{The proof for this statement and further discussion can be found in the technical report \cite{optTR}.}\techreport{We discuss this statement and describe our proof for it in the appendix, Section \ref{sec:appendix-extensions-variable}.}

\subsection{Worker classes}
So far we have assumed that all workers are identical, in that they draw their answers from the same response probability matrix, a strong assumption that does not hold in general. Although we could argue that different worker matrices could be aggregated into one average probability matrix that our previous approach discovers, if we have fine-grained knowledge about workers, we would like to exploit it. In this section we consider the setting where there are two of classes of workers,
%This is a very natural setting in many situations where although individual workers cannot be distinguished, we are given access to pools of workers where different pools have similar characteristics within them. One very common example is that of using
\emph{expert} and \emph{regular} workers to evaluate the same set of items. %In general, experts are characterized as being more expensive, resource-intensive, and reliable than regular workers. %Recall our example from Section \ref{intro} of finding images of Barack Obama. Suppose our journalist decides to ask workers on Mechanical Turk to evaluate several images while also asking their colleagues to evaluate a few. They know that
We discuss the generalization to larger numbers of worker classes below.

We now model worker behavior as two different response probability matrices, the first corresponding to expert workers who have low error rates, and the second corresponding to regular workers who have higher error rates. Our problem now becomes that of estimating the items' true values in addition to both of the response probability matrices.
For this discussion, we consider the filtering problem; %and extend it to encapsulate worker classes.
a similar analysis can be applied to the rating case.

Again, we extend our ideas of bucketizing and dominance to this setting. Let $(y_e,n_e,y_r,n_r)$ be the bucket representing all items that receive $y_e$ and $n_e$ responses of ``1'' and ``0'' respectively from experts, and $y_r$ and $n_r$ responses of ``1'' and ``0'' respectively from regular workers. A dominance partial ordering can be defined using the following rules. An item (respectively bucket) with response set $B_1=(y_e^1,n_e^1,y_r^1,n_r^1)$ dominates an item (respectively bucket) with response set $B_2=(y_e^2,n_e^2,y_r^2,n_r^2)$ if and only if one of the following is satisfied:
%\begin{enumerate}
\begin{denselist}
\item $B_1$ sees more responses of ``1'' and fewer responses of ``0'' than $B_2$. That is, $(y_e^1\geq y_e^2)\land(y_r^1\geq y_r^2)\land(n_e^1\leq n_e^2)\land(n_e^1\leq n_e^2)$ where at least one of the inequalities is strict.
\item $B_1$ and $B_2$ see the same number of ``1'' and ``0'' responses in total, but more experts respond ``1'' to $B_1$ and ``0'' to $B_2$. That is, $(y_e^1+y_r^1=y_e^2+y_r^2)\land(n_e^1+n_r^1=n_e^2+n_r^2)\land(y_e^1\geq y_e^2)\land(n_e^1\leq n_e^2)$ where at least one of the inequalities is strict.
\end{denselist}
%\end{enumerate}
As before, we consider only the set of mappings that assign all items in a bucket the same value while preserving the dominance relationship, that is, dominating buckets get at least as high a value as dominated buckets. %Here, we are assuming that worker answers come annotated with the class that they belong to.

Note that the second dominance condition above leverages the assumption that experts have smaller error probabilities than regular workers. If we were just given two classes of workers with no information about their response probability matrices, we could only use the first dominance condition.
%The finer grained information we have about the worker classes, the more dominance conditions we can construct, which reduces the number of dominance-consistent mappings that we need to consider. Overall, our framework is flexible and adaptable to different granularities of prior knowledge.
In general, having more information about the error probabilities of
worker classes allows us to construct stronger dominance conditions,
which in turn reduces the number of dominance-consistent mappings. This
property allows our framework to be flexible and adaptable to different granularities of prior knowledge.

\begin{figure}
\vspace{-5pt}
\centering
\includegraphics[scale=0.38]{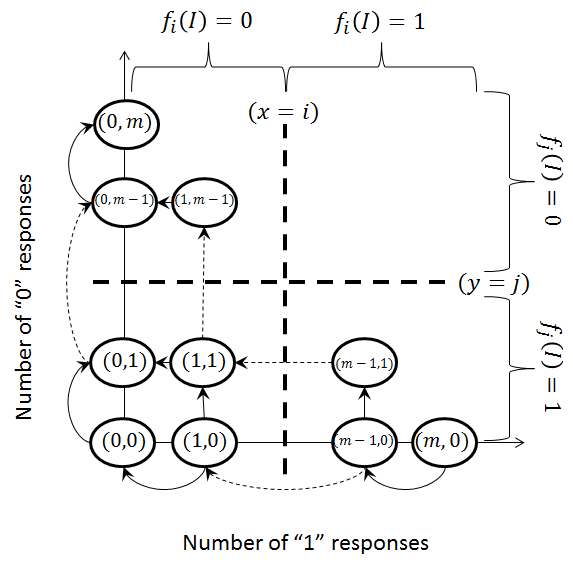}
\vspace{-10pt}
\caption{Variable number of responses}
\label{fig:ext-1}
\vspace{-20pt}
\end{figure}

While this extension is reasonable when the number of distinct worker classes is small, it is impractical to generalize it to a large number of classes. One heuristic approach to tackling the problem of a large number of worker classes, or independent workers, could be to divide items into a large number discrete groups and assign a small distinct set of workers to evaluate each group of items. We then treat and solve each of the groups independently as a problem instance with a small number worker classes.  More efficient algorithms for this setting is a topic for future work.

\section{Conclusions}% and Future Work}
\label{conclusions}
We have taken a first step towards finding a global maximum likelihood solution to the problem of jointly estimating the item ground truth, and worker quality, in crowdsourced filtering and rating tasks. %the problem of ordinal labeling using the
Given worker ratings on a set of items (binary in the case of filtering), we show that the problem of jointly estimating the ratings of items and worker quality can be split into two independent problems. We use a few key, intuitive ideas to first find a global maximum likelihood mapping from items to ratings, thereby finding the most likely ground truth. We then show that the worker quality, modeled by a common response probability matrix, can be inferred automatically from the corresponding maximum likelihood mapping. We develop a novel pruning and search-based approach, in which we greatly reduce the space of (originally exponential) potential mappings to be considered, and prove that an exhaustive search in the reduced space is guaranteed to return a maximum likelihood solution.%Our approach relies on simplifying the space of potential

We performed experiments on real and synthetic data to compare our algorithm against an Expectation-Maximization based algorithm. We show that in spite of being optimized for the likelihood of mappings, our algorithm estimates the ground truth of item ratings and worker qualities with high accuracy, and performs well over a number of comparison metrics.

Although we assume throughout most of this paper %that worker identities are unknown and 
that all workers draw their responses independently from a common probability matrix, we generalize our approach to the cases where different worker classes draw their responses from different matrices. %and where different items may receive different numbers of responses. 
Likewise, we assume a fixed number of responses for each item, but we can generalize to the case where different items may receive different numbers of responses.

It should be noted that although our framework generalizes to these extensions, including the case where each worker has an independent, different quality, the algorithms can be inefficient in practice. We have not considered the problem of item difficulties in this paper, assuming that workers have the same quality of responses on all items.
As future work, we hope that the ideas described in this paper can be built upon to design efficient algorithms that find a global maximum likelihood mapping under more general settings. %where every worker is distinct. As our framework uses very simple intuitive observations to prune the space of potential  mappings, it leaves room for more sophisticated algorithms to be built on top of the reduced space. We leave the problem of scaling our approach to harder, more generalized settings as future work. 
%\nocite{*}

%\fontsize{9.5pt}{10.5pt} \selectfont

\papertext{\newpage}
{\scriptsize
\bibliographystyle{abbrv}
\bibliography{mab,main,agprefs}
}

\appendix
%!TEX root=Parameter Estimation.tex

\section{Filtering}
\label{sec:appendix-filter}
\subsection{Experiments}
\label{sec:appendix-filter-exp}
\smallskip
\noindent {\bf Metrics.} 
Earth-movers distance, or EMD, is a metric function that captures how similar two probability distributions are. Intuitively, if the two distributions are represented as piles of sand, EMD is a measure of the minimum amount of sand that needs to be shifted to make the two piles equal. In our problem, the worker response matrix $p$ can be represented as two probability distributions corresponding to $p(i,1)$ and $p(i,0)$, that is the probability distributions of worker responses given that the true value of an item is 1 and 0 respectively. We compute the EMD of $p(i,1)$ from $p_{\text{true}}(i,1)$ and $p(i,0)$ from $p_{\text{true}}(i,0)$ and record their sum as the EMD ``score'' of the algorithm that predicts $p$. Since the EMD between $p(i,j)$ and $p_{\text{true}}(i,j)$ lies between $[0,1]$, our EMD score that sums the individual EMDs for $i=0,1$ lies in $[0,2]$.

We also compute a similar score using the Jensen-Shannon divergence (JSD), which is another standard metric for measuring the similarity between two probability distributions. We compute the JSD between $p(i,1)$ and $p_{\text{true}}(i,1)$, and $p(i,0)$ and $p_{\text{true}}(i,0)$. As with the EMD based score, we compute the sum of these two JSD values and use it as our comparison metric.

\smallskip
\noindent {\bf Additional results.}
We now present some additional experimental plots comparing our algorithm against the different EM instances for the metrics of likelihood, EMD based score, JSD based score, and fraction of items predicted incorrectly, similar to Section \ref{sec:filter-exp-simulated}.

Figures \ref{fig:filter_frac5} and \ref{fig:filter_jsd5} plot the fraction of items predicted incorrectly and JSD score respectively for a selectivity of 0.5. We see that both our algorithm and EM perform comparably on these metrics and give high accuracy. Figure \ref{fig:filter_jsd7} plots the JSD score for a selectivity of 0.7. We observe that here our algorithm outperforms the aggregated $EM(*)$ algorithm, while $EM(1)$ is comparable.
\begin{figure*}[!t]
%\vspace{-10pt}
%\centering
\subfigure{
\includegraphics[scale=0.3]{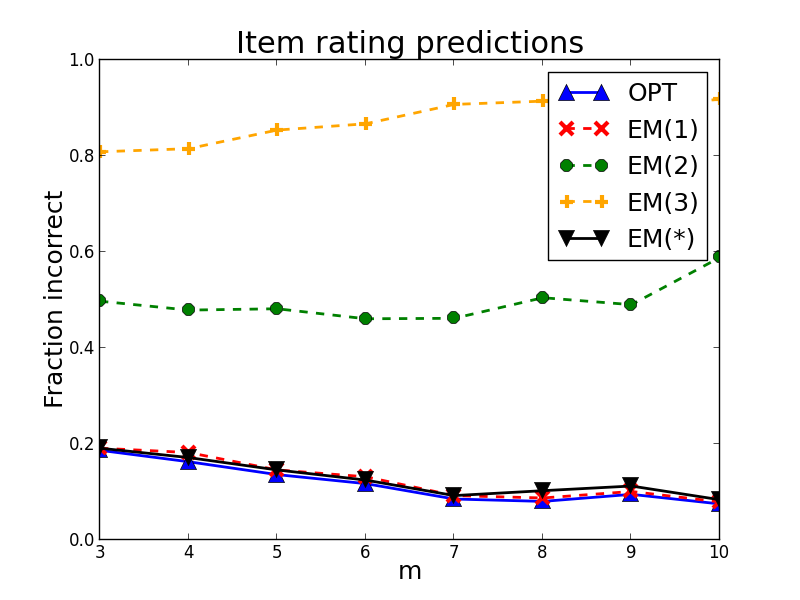}
\label{fig:filter_frac5}
}
\hspace{-15pt}
\subfigure{
\includegraphics[scale=0.3]{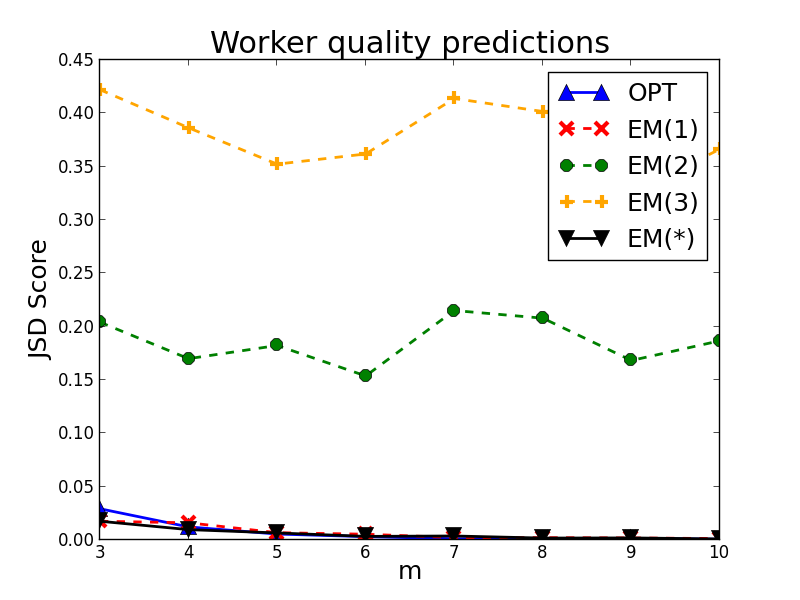}
\label{fig:filter_jsd5}
}
\hspace{-15pt}
\subfigure{
\includegraphics[scale=0.3]{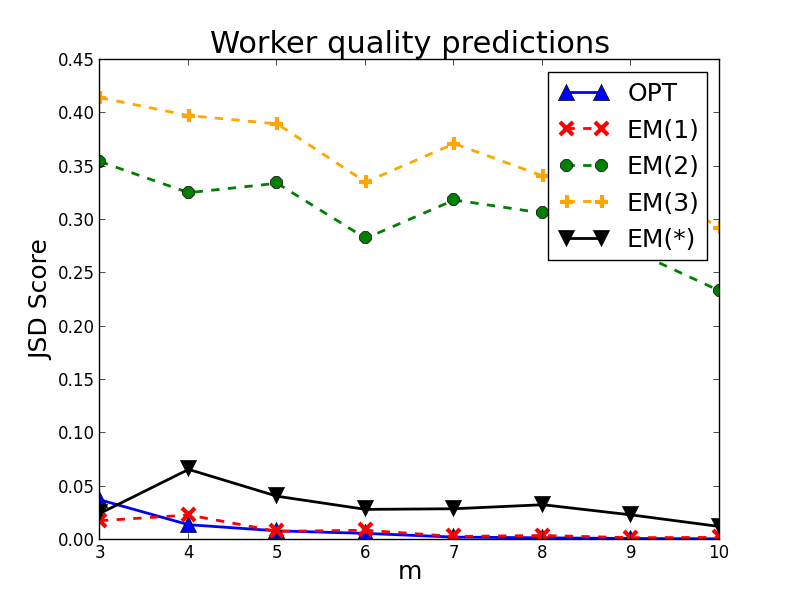}
\label{fig:filter_jsd7}
}
\vspace{-20pt}
\caption{Synthetic Data Experiments: (a)Fraction Incorrect, $s=0.5$ (b) JSD Score, $s=0.5$ (c) JSD Score, $s=0.7$}
%\vspace{-15pt}
\end{figure*}

We also generate synthetic data with the ground truth having a selectivity of 0.9, that is 90\% of items have a true value of 1 and 10\% have a true value of 0. We observe, from Figures \ref{fig:filter_frac9} and \ref{fig:filter_jsd9}, that truth our algorithm outperforms $EM(*)$, but does worse than $EM(1)$ over this highly skewed ground. We explain this effect at a high level with the following intuitive example: suppose all items had a true value of 1 and workers had a false negative error rate of 0.2. Then, we expect 80\% of all worker responses to be 1, and the remaining 20\% to be 0. Now for this worker response set, the pair $s=1,e_1=0.2$ is less likely than other less ``extreme'' solutions, $s=0.9,e_1=0.1$ for example. As a result, while $EM(1)$ readily converges to such extreme solution points, $OPT$ and $EM(*)$ find more likely solutions which for such highly skewed instances turn out to be less accurate. In practice it is often not possible to predict when a given dataset will be skewed or ``extreme''. Given such information, we can tune the $EM$ and $OPT$ algorithms to account for the skewness and find better solutions.
\begin{figure*}[!t]
%\vspace{-10pt}
%\centering
\subfigure{
\includegraphics[scale=0.3]{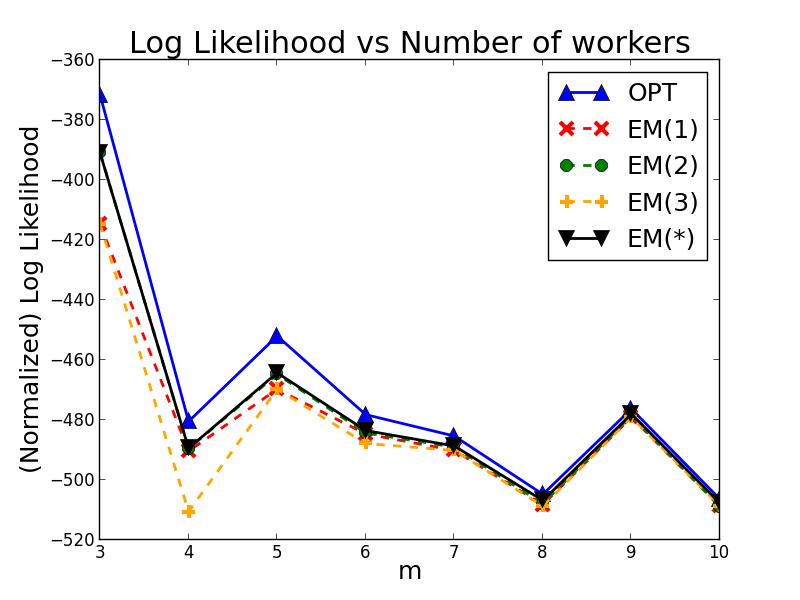}
\label{fig:like_frac9}
}
\hspace{-15pt}
\subfigure{
\includegraphics[scale=0.3]{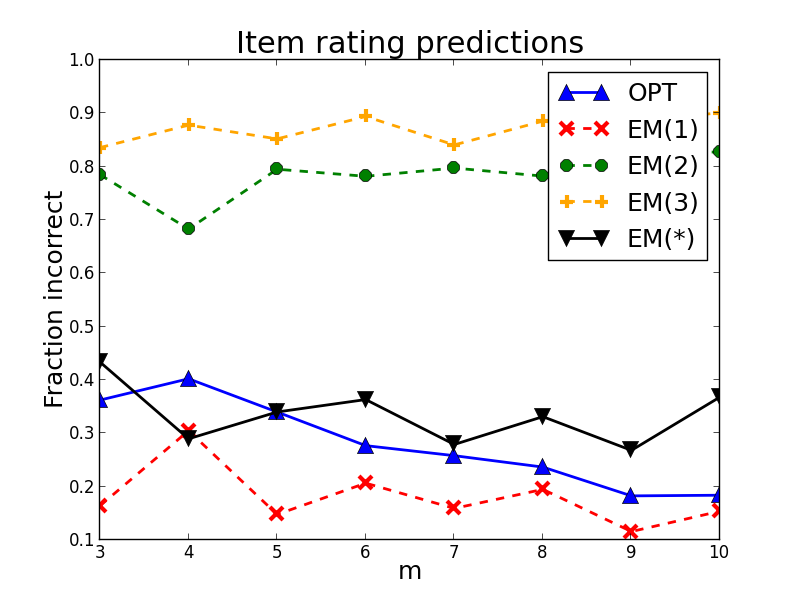}
\label{fig:filter_frac9}
}
\hspace{-15pt}
\subfigure{
\includegraphics[scale=0.3]{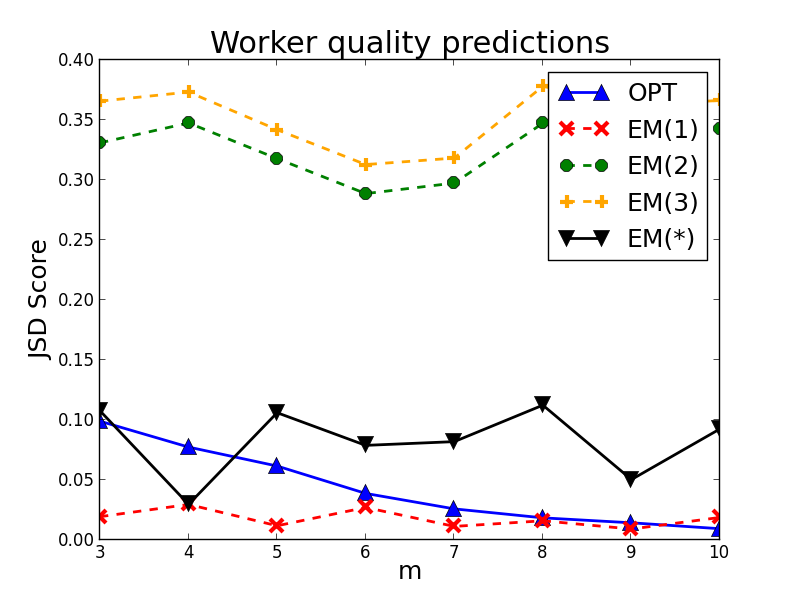}
\label{fig:filter_jsd9}
}
\vspace{-20pt}
\caption{Synthetic Data Experiments: (a) Likelihood, $s=0.9$ (b)Fraction Incorrect, $s=0.9$ (c) JSD Score, $s=0.9$}
%\vspace{-15pt}
\end{figure*}

\section{Rating}
\label{sec:appendix-rating}
\subsection{Experiments}
\label{sec:appendix-rating-experiments}
\smallskip
\noindent {\bf Metrics.} 
In this section, we present results over different input ground truth distributions over the same experimental setup described in Section \ref{sec:rating-exp}. Recall that there we describe our results for the case where items are equally divided across the rating values, that is, one-third each of the items have true ratings $1$, $2$ and $3$ respectively.
%\begin{figure}
%\centering
%\includegraphics[scale=0.3]{rating-emd.png}
%%\vspace{-15pt}
%\caption{EMD Score, $R=3$}
%\label{fig:rating-emd}
%%\vspace{-15pt}
%\end{figure}

\begin{figure*}[!t]
%\vspace{-10pt}
%\centering
\subfigure{
\includegraphics[scale=0.3]{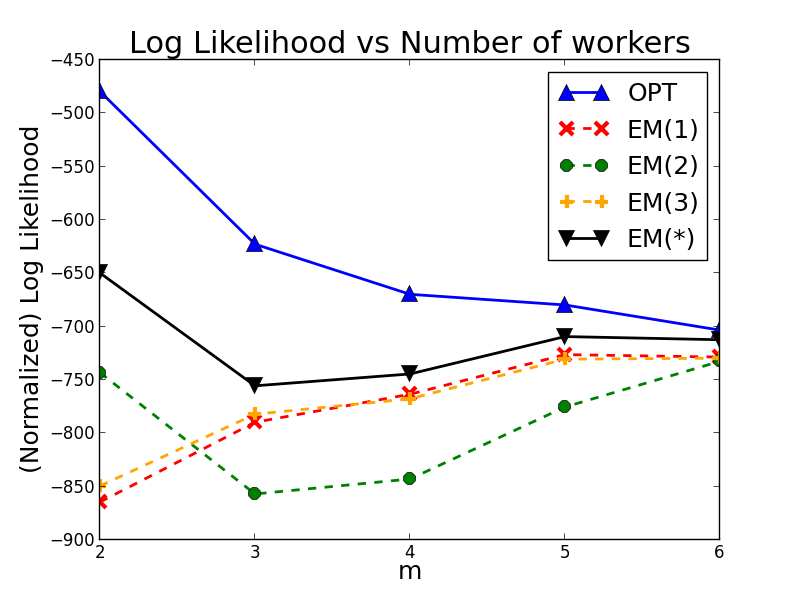}
\label{fig:rating-like2}
}
\hspace{-15pt}
\subfigure{
\includegraphics[scale=0.3]{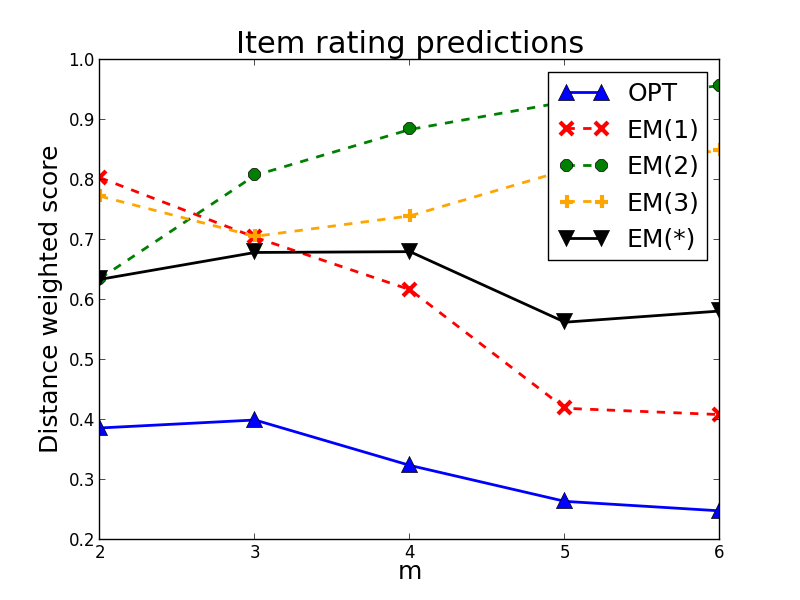}
\label{fig:rating-dist2}
}
\hspace{-15pt}
\subfigure{
\includegraphics[scale=0.3]{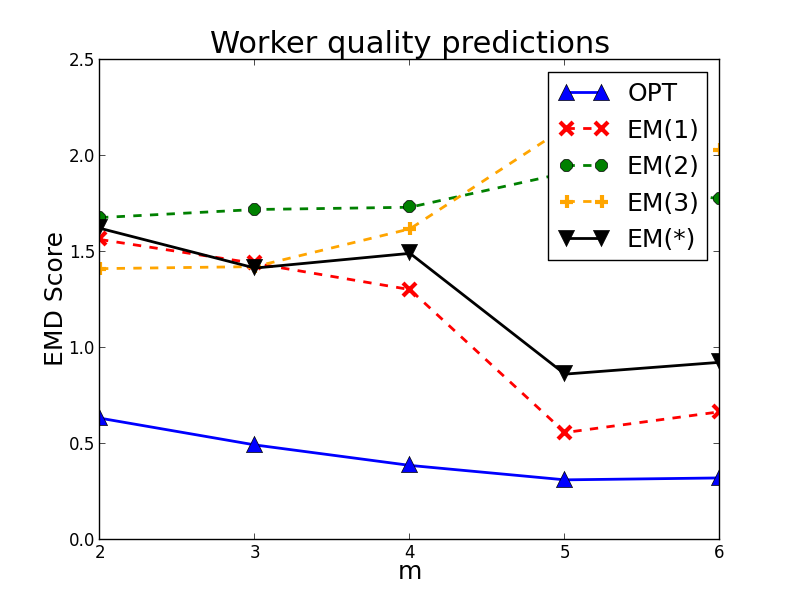}
\label{fig:rating-emd2}
}
\vspace{-20pt}
\caption{Synthetic Data Experiments: (a) Likelihood, $s=2$ (b)Distance Weighted Score, $s=2$ (c) EMD Score, $s=2$}
\label{fig:rating-2}
%\vspace{-15pt}
\end{figure*}

\begin{figure*}[!t]
%\vspace{-10pt}
%\centering
\subfigure{
\includegraphics[scale=0.3]{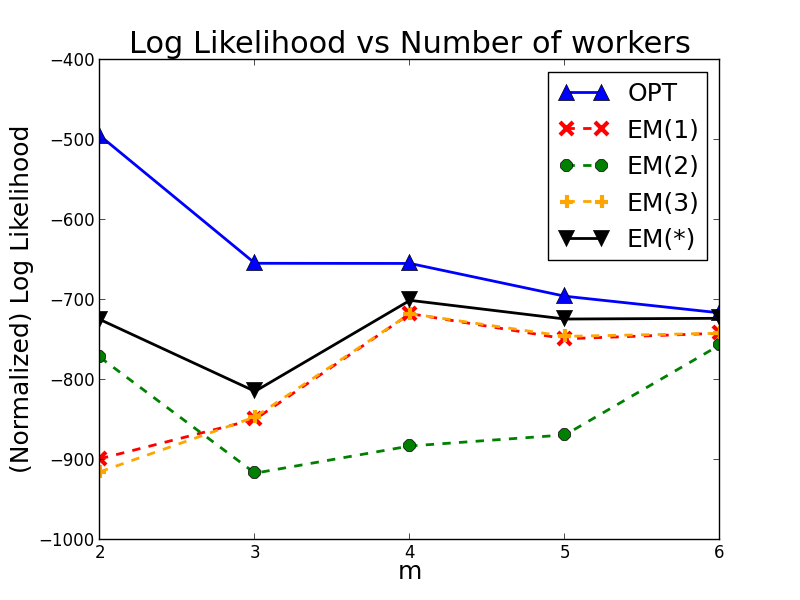}
\label{fig:rating-like3}
}
\hspace{-15pt}
\subfigure{
\includegraphics[scale=0.3]{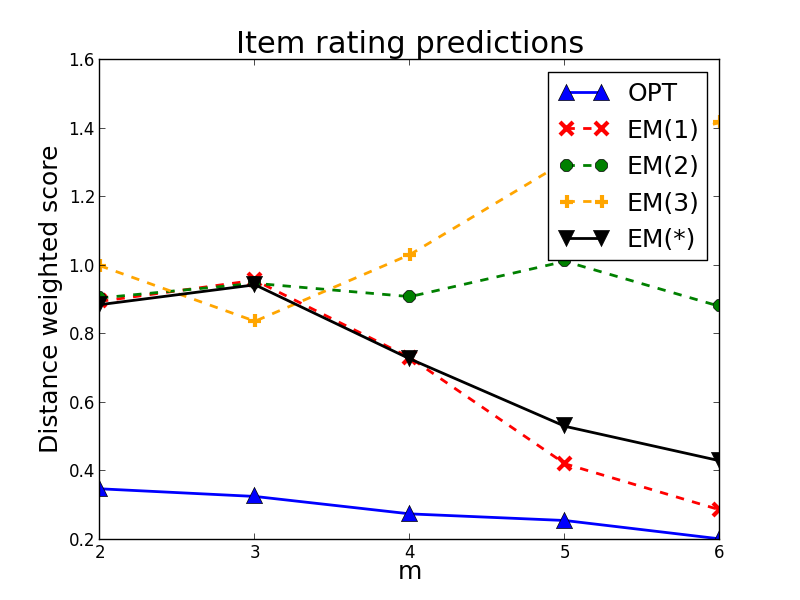}
\label{fig:rating-dist3}
}
\hspace{-15pt}
\subfigure{
\includegraphics[scale=0.3]{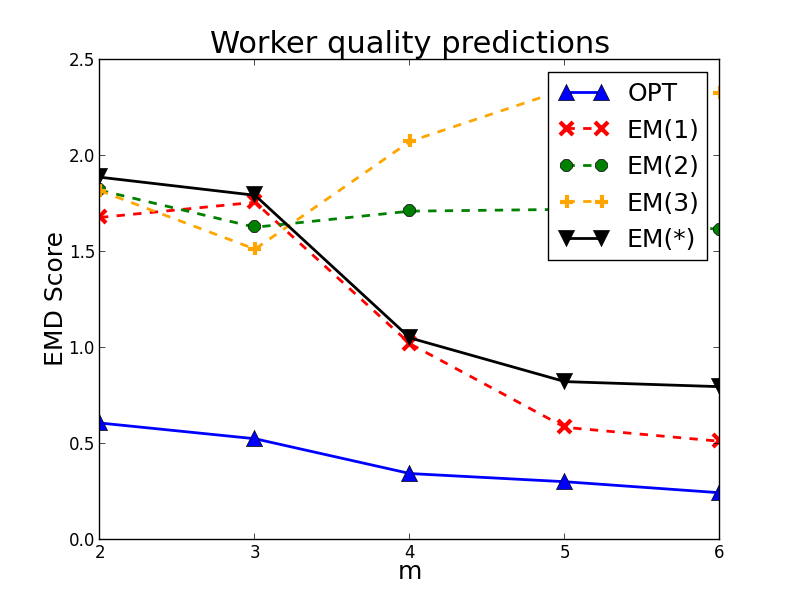}
\label{fig:rating-emd3}
}
\vspace{-20pt}
\caption{Synthetic Data Experiments: (a) Likelihood, $s=3$ (b)Distance Weighted Score, $s=3$ (c) EMD Score, $s=3$}
\label{fig:rating-3}
%\vspace{-15pt}
\end{figure*}

In addition to the distance-weighted item prediction score, and likelihood metrics described in Section \ref{sec:rating-exp}, we also run experiments on our synthetic data using the EMD based score described in Section \ref{sec:appendix-filter-exp}. Note that here the EMD score is based on the sum of $R=3$ pairwise EMD values corresponding to each column of the $3\times 3$ response probability matrix.

In Figure \ref{fig:rating-2} we plot experiments where 20\% of the items have ground truth rating 1, 60\% have ground truth rating 2 and 20\% have ground truth rating 3. We respresent this ground truth distribution, or selectivity vector by the notation $s=2$. In Figure \ref{fig:rating-3} ($s=3$) we plot experiments where 40\% of the items have ground truth rating 1, 20\% have ground truth rating 2 and 40\% have ground truth rating 3.

We observe that for all these experiments, the results are very similar to those seen in Section \ref{sec:rating-exp}, Figures \ref{fig:rating-like} and \ref{fig:rating-dist_wtd}. We observe that our algorithm finds more likely mappings (Figures \ref{fig:rating-like2}, \ref{fig:rating-like3}), predicts item ground truth ratings with higher accuracy (Figures \ref{fig:rating-dist2}, \ref{fig:rating-dist3}) and obtains a better estimate for the worker response probability matrix (Figures \ref{fig:rating-emd2}, \ref{fig:rating-emd3}) than all EM instances.

\section{Extensions}
\label{sec:appendix-extensions}
\subsection{Variable number of responses}
\label{sec:appendix-extensions-variable}
In this section, we calculate the number of dominance-consistent mappings for the filtering problem where different items can each receive a different number of worker responses. Let $m$ be the maximum number of worker responses that any item receives. Recall from Section \ref{sec:extensions-variable}that we can represent the set of all possible item response sets in the dominance-DAG shown in Figure \ref{fig:ext-1}.

Let $f$ be any dominance-consistent mapping under this setting. Let $(i,j)$ be a response set with $i$ responses of 1 and $j$ responses of $0$ such that $f(i,j)=1$. Then, by our dominance constraint, we know that $f(i+\Delta i, j+\Delta j)=1\forall \Delta i\geq 0, \Delta j\geq 0$.  We represent this figuratively in Figure \ref{fig:ext-2}. If $f(i,j)=1$, then all response sets, or points, in the shaded area also get mapped to a value of 1. 

\begin{figure}
%\vspace{-5pt}
\centering
\includegraphics[scale=0.45]{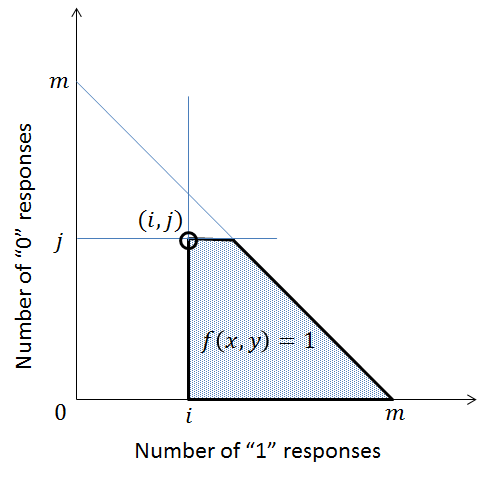}
\vspace{-10pt}
\caption{Dominance constraint}
\label{fig:ext-2}
%\vspace{-20pt}
\end{figure}

Now, by applying the dominance constraint as shown in Figure \ref{fig:ext-2} to every $(i,j)$ such that $f(i,j)=1$, we can show that $f$ can now be described intuitively by its ``boundary''. We demonstrate this intuition in Figure \ref{fig:ext-3}. Every point, $(i,j)$, on $f$'s boundary satisfies $f(i,j)=1$. Additionally, every point ``within'' the boundary, every point that is to the right of and below the boundary gets mapped to a value of $1$ under $f$. Finally, every point that is not on or within the boundary gets mapped to a value of $0$ under $f$. 
\begin{figure}
%\vspace{-5pt}
\centering
\includegraphics[scale=0.45]{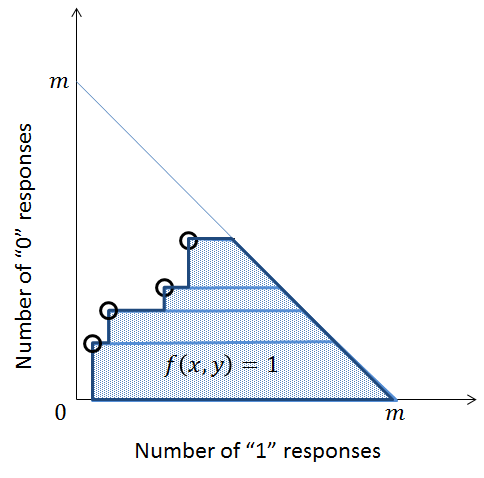}
\vspace{-10pt}
\caption{Dominance constraint}
\label{fig:ext-3}
%\vspace{-20pt}
\end{figure}
It follows from the dominance constraint that every dominance-consistent mapping, $f$, can be represented by a unique such continuous boundary. Furthermore, it is easy to see that any such boundary satisfies three conditions:
\begin{denselist}
\item Its leftmost (corner) point lies on one of the axes.
\item Its topmost (corner) point lies on the line $x+y=m$.
\item If $(x_1,y_1)$ and $(x_2,y_2)$ lie on the boundary, then $x_1\leq x_2\Leftrightarrow y_1\leq y_2$.
\end{denselist}
Intuitively the above three conditions give us a constructive definition for any dominance-consistent mapping's boundary. Every dominance-consistent mapping can be constructed uniquely as follows: (1) Choose a left corner point lying on one of the axes. (2) Choose a topmost corner point (necessarily above and to the right of the first corner point) lying on the line $x+y=m$. (3) Finally, define the boundary as a unique grid traversal from the left corner to the top corner where you are allowed to extend the boundary only to the right or upwards. Each such boundary corresponds to a unique dominance-consistent mapping where every point on or under the boundary is mapped to 1 and every other point is mapped to 0. Furthermore, every dominance-consistent mapping has a unique such boundary.

Therefore our problem of counting the number of dominance-consistent mappings for this setting reduces to counting the number of such boundaries. We use our constructive definition for the boundary to compute this number. First, suppose the leftmost corner point, $L=(p,0), 0\leq p\leq m$, lies on the $x$-axis (we can calculate similarly for $(0,q)$). Now, the topmost corner point lies to the right of the first corner point, and on the line $x+y=m$. Therefore it is of the form $T=(p+i,m-(p+i))$ for some $0\leq i\leq m-p$. The number of unique grid traversals (respectively boundaries) from $L$ to $T$ is given by ${m-p \choose i}$. Combining, we have the number of unique boundaries that have their left corner on the $x$-axis is $\overset{m}{\underset{p=0}\sum} \overset{m-p}{\underset{i=0}\sum }{m-p \choose i}=\overset{m}{\underset{p=0}\sum}2^{m-p}=2^{m+1}-1$. Calculating similarly for boundaries that start with their leftmost corner on the $y$-axis ($(0,q), 1\leq q\leq m$) and including the empty boundary (corresponding to the mapping where all items get assigned a value of 0), we get an additional $2^m$ boundaries. Therefore, we conclude that there are $O(2^m)$ such boundaries, corresponding to $O(2^m)$ dominance-consistent mappings. It should be noted that although this is exponential in the maximum number of worker responses to an item, typical values of $m$ are small enough that all mappings can very easily be enumerated and evaluated.

\end{document}